\documentclass[letterpaper,11pt]{article}

\PassOptionsToPackage{table,xcdraw}{xcolor}
\usepackage[utf8]{inputenc}
\usepackage[T1]{fontenc}
\usepackage{amsmath, amsthm, amssymb, thm-restate}
\usepackage{dsfont}
\usepackage{algorithmicx}
\usepackage[ruled,vlined,linesnumbered]{algorithm2e}
\usepackage{setspace}
\usepackage{mathtools}
\usepackage[numbers]{natbib}
\usepackage{comment} 
\usepackage[most]{tcolorbox} 
\usepackage{xfrac}
\usepackage{hyperref}
\usepackage{multirow}
\usepackage{caption}
\usepackage{bm}
\usepackage{float}
\usepackage{newfloat}
\usepackage{enumitem}
\usepackage{dblfloatfix} 
\usepackage{wrapfig}
\usepackage{tikz}
\usepackage{csquotes}

\usepackage{fancyhdr}

\SetCommentSty{mycommfont}
\let\oldnl\nl
\newcommand{\nonl}{\renewcommand{\nl}{\let\nl\oldnl}}
\newcommand{\dmax}{\Delta_{\mathsf{max}}}
\usepackage[margin=1in]{geometry}

\renewcommand{\epsilon}{\varepsilon}
\newcommand{\eps}{\varepsilon}

\newcommand*{\E}{\mathbf{E}}

\newcommand{\hiddencomment}[1]{}

\newcommand{\mc}[1]{\ensuremath{\mathcal{#1}}}

\newcommand{\Bernoulli}{\text{Bernoulli}}
\newcommand{\one}{\mathds{1}}
\renewcommand{\dh}{\widehat{d}}
\newcommand{\RGMM}{\textsc{RGMM}}

\newcommand{\GMM}{\textsc{GMM}}
\newcommand{\False}{\textsc{False}}
\newcommand{\True}{\textsc{True}}

\newcommand{\R}{\mathbb{R}}
\newcommand{\norm}[1]{\left\lVert #1 \right\rVert}

\DeclareMathOperator{\poly}{\mathsf{poly}}

\usepackage[noabbrev,nameinlink]{cleveref}
\crefname{lemma}{Lemma}{Lemmas}
\crefname{algorithm}{Algorithm}{Algorithms}
\crefname{theorem}{Theorem}{Theorems}
\crefname{property}{Property}{Properties}
\crefname{claim}{Claim}{Claims}
\crefname{definition}{Definition}{Definitions}
\crefname{observation}{Observation}{Observations}
\crefname{proposition}{Proposition}{Propositions}
\crefname{assumption}{Assumption}{Assumptions}
\crefname{line}{Line}{Lines}
\crefname{figure}{Figure}{Figures}
\creflabelformat{property}{(#1)#2#3}
\crefname{equation}{}{}
\crefname{section}{Section}{Sections}
\crefname{appendix}{Appendix}{Appendices}

\newtheorem{theorem}{Theorem}
\newtheorem{lemma}{Lemma}[section]
\newtheorem{proposition}[lemma]{Proposition}
\newtheorem{corollary}[lemma]{Corollary}

\newtheorem{definition}[lemma]{Definition}
\newtheorem{claim}[lemma]{Claim}
\newtheorem{fact}[lemma]{Fact}

\newtheorem{remark}[lemma]{Remark}
\newtheorem{assumption}[lemma]{Assumption}

\newcommand{\ceil}[1]{{\left\lceil{#1}\right\rceil}}
\newcommand{\floor}[1]{{\left\lfloor{#1}\right\rfloor}}

\newcommand{\card}[1]{\lvert#1\rvert}

\DeclareMathOperator*{\Exp}{\ensuremath{{\mathbb{E}}}}
\DeclareMathOperator*{\Prob}{\ensuremath{\textnormal{Pr}}}
\renewcommand{\Pr}{\Prob}

\definecolor{mylightgray}{RGB}{240,240,240}

\algnewcommand{\IIf}[2]{\textbf{if} #1 \textbf{then} #2}
\algnewcommand{\EndIIf}{\unskip\ \algorithmicend\ \algorithmicif}
\algnewcommand{\IElse}[1]{\textbf{else} #1}

\newenvironment{graytbox}{
\par\addvspace{0.1cm}
\begin{tcolorbox}[width=\textwidth,
                  boxsep=5pt,
                  left=1pt,
                  right=1pt,
                  top=2pt,
                  bottom=2pt,
                  boxrule=1pt,
                  arc=0pt,
                  colback=mylightgray,
                  colframe=black,
                  ]
}{
\end{tcolorbox}
}

\newenvironment{whitetbox}{
\par\addvspace{0.1cm}
\begin{tcolorbox}[width=\textwidth,
                  boxsep=5pt,
                  left=1pt,
                  right=1pt,
                  top=2pt,
                  bottom=2pt,
                  boxrule=1pt,
                  arc=0pt,
                  colframe=black,
                  colback=white
                  ]
}{
\end{tcolorbox}
}

\newcounter{algCounter}

\allowdisplaybreaks

\definecolor{mygreen}{RGB}{10,110,230}
\definecolor{myred}{RGB}{10,110,230}

\hypersetup{
     colorlinks=true,
     citecolor= mygreen,
     linkcolor= myred,
     urlcolor=mygreen
}

\makeatletter
\renewcommand{\paragraph}{%
  \@startsection{paragraph}{4}%
  {\z@}{10pt}{-1em}%
  {\normalfont\normalsize\bfseries}%
}
\makeatother

\makeatletter
\patchcmd{\@algocf@start}
  {-1.5em}
  {0pt}
  {}{}
\makeatother

\title{Streaming Algorithms for Monotonicity Testing}

\date{}

\author{
Amir Azarmehr\thanks{Northeastern University, \{\texttt{azarmehr.a}, \texttt{s.behnezhad}, \texttt{ghafari.m}\}\texttt{@northeastern.edu}.  A.A., S.B., and A.G. were supported in part by NSF CAREER Award CCF-2442812, a Sloan Fellowship, and a Google Research Award.} \and
Soheil Behnezhad\footnotemark[1] \and
Lily Chung\thanks{MIT CSAIL, \{\texttt{lkdc}, \texttt{jlange}\}\texttt{@mit.edu}, \texttt{ronitt@csail.mit.edu}. L.C. and J.L. were supported by the Big George fellowship.  L.C., J.L., and R.R. were supported by NSF TRIPODS Award DMS-2022448 and NSF Award CCF-2310818.} \and 
Alma Ghafari\footnotemark[1] \and
Jane Lange\footnotemark[2] \and 
Ronitt Rubinfeld\footnotemark[2]
}

\begin{document}

\maketitle

\thispagestyle{empty}

\begin{abstract}
\setlength{\parskip}{0.2cm}
 Consider a poset---or equivalently an $n$-vertex DAG $G=(V, \vec{E})$---and a boolean function $f: V \rightarrow \{0, 1\}$ on its vertex set. We say $f$ is {\em monotone} if $f(u) \leq f(v)$ for all $(u, v) \in \vec{E}$. While there is extensive literature on the {\em query complexity} of testing monotonicity, we focus instead on the {\em space complexity} and initiate the study of this problem in the {\em streaming} setting. Namely, the edges of $G$ arrive in an arbitrary order, and the goal is to estimate distance to monotonicity of a given function $f$ using $\widetilde{O}(n)$ space. Note that while this space allows receiving and storing $f$, it is much smaller than the input graph $G$ which could have up to $\Omega(n^2)$ edges.

 Our main result is an algorithm that $(1+\epsilon)$-approximates distance to monotonicity in $\sqrt{n}^{1+o(1)}$ passes. We also prove that this is the best pass-complexity one can hope for, for any $O(1)$-approximation, short of improving the state-of-the-art streaming algorithm for $st$-reachability, which is a very well-studied problem.

 On the technical side, our algorithm approximates the size of maximum matching in (a subgraph of) the transitive closure of $G$. While the maximum matching problem has received significant attention in the streaming setting, the fact that we are computing it in the transitive closure requires very different ideas. In fact, a main contribution of our work is to connect sublinear time algorithms for estimating the maximum matching size to the streaming setting for the first time. While existing off-the-shelf sublinear time algorithms only result in an $n\sqrt{n}^{1+o(1)}$ pass algorithm in our setting, we show how to significantly improve upon them by allowing stronger queries (such as vertex and subset queries) that can be implemented just as efficiently as more standard adjacency matrix and list queries for our problem.

\end{abstract}

\clearpage
\tableofcontents
\thispagestyle{empty}
\clearpage

\setcounter{page}{1}

\section{Introduction}

\newcommand{\dmon}[1]{\ensuremath{d_{\mathsf{mon}}(#1)}}
\newcommand{\streach}[0]{\ensuremath{\mathsf{st\textsf{-}reach}_n}}

 Consider a poset---or equivalently an $n$-vertex directed acyclic graph (DAG) $G=(V, \vec{E})$. Suppose we are also given a boolean function $f: V \rightarrow \{0, 1\}$ defined on the vertex set of this graph. Function $f$ is said to be {\em monotone} if $f(u) \leq f(v)$ for all $(u, v) \in \vec{E}$. We study algorithms that efficiently determine the {\em distance to monotonicity} of $f$, denoted $\dmon{f}$, which is defined as the minimum number of entries of $f$ that have to change in order to make $f$ monotone.
 
 Monotonicity testing has been studied extensively in the literature (see \cref{sec:relatedwork}). However, nearly all these works study the {\em query complexity} of this problem. That is, they assume that the graph $G$ is entirely known in advance and can be accessed for free, and study how many queries to $f$ are needed in order to estimate its distance to monotonicity. Our focus instead is on massive inputs, where the graph $G$ cannot be stored and randomly accessed by the algorithm.
 
  More specifically, we study {\em streaming} algorithms for approximating distance  to monotonicity: The edges of $G$ arrive one by one in a stream, and the algorithm has space much smaller than this input. As standard in the literature of streaming graph algorithms, we allow $\widetilde{O}(n)$ space\footnote{Here and throughout the paper we use $\widetilde{O}(f)$ to suppress $\poly(\log n)$ factors.}, which is enough to receive and store the function $f$ entirely, but is much smaller than the graph $G$ which could include $\Omega(n^2)$ edges. 

Our main result in this paper is the following algorithm:

\begin{graytbox}
    \begin{theorem}\label{thm:main}
        For any $\epsilon > 0$, there is a randomized streaming algorithm that takes $\sqrt{n}^{1+o(1)} \cdot \poly(1/\eps)$ passes over the stream, uses $\widetilde{O}(n/\eps^2)$ space, and w.h.p. $(1+\eps)$-approximates $\dmon{f}$.
    \end{theorem}
\end{graytbox}

At first glance, the $\sqrt{n}^{1+o(1)}$ pass-complexity might appear prohibitive. However, we prove that this is the best bound one can hope for short of a breakthrough in streaming algorithms. 
Specifically, we show in \cref{thm:lb} that any $p$-pass $O(1)$-approximation of \dmon{f} also solves $st$-reachability in $p$ passes and $\widetilde{O}(n)$ space.
Since the fastest known streaming $st$-reachability algorithm uses $\sqrt{n}^{1+o(1)}$ passes \cite{JLS19, AJJST22}, this shows that any faster algorithm for approximately determining the distance to monotonicity would also improve the pass-complexity of streaming $st$-reachability.

In fact, our algorithm uses the streaming version of \cite{JLS19} directly to compute a \emph{shortcut set}, a set of edges which can be added to the graph to reduce its diameter without changing the reachability relation on vertices.
Our algorithm's pass complexity is (ignoring $\poly(\log n)$ factors) equivalent to the diameter of this shortcut set.
Therefore, as long as the fastest known streaming $st$-reachability algorithms continue to work by computing a shortcut set, the pass-complexity of \cref{thm:main} will match those algorithms.





\subsection{Main Technique: Query Complexity of Maximum Matching Size}

The key contribution of our work is to tie the streaming complexity of monotonicity testing to {\em sublinear time} algorithms for estimating the size of maximum matching. The latter problem has been studied extensively on its own over the last two decades, and has a rich literature with numerous applications (see \cref{sec:relatedwork}). However, to our knowledge this is the first application of this problem in the streaming setting. Additionally, as we soon discuss, our work motivates the study of the query complexity of the maximum matching problem in alternative models, where the queries differ from the standard adjacency-list or adjacency-matrix queries. 

\paragraph{Background: Monotonicity and matchings.} The connection between the 
maximum matching problem and monotonicity of boolean functions dates back  to the early monotonicity testing
work of \cite{DodisGLRRS99}. 
Let $G_f$ be the {\em violation graph} defined on the same vertex set $V$, which has an edge $\{u, v\}$ whenever $u$ can reach $v$ but $f(u) > f(v)$. Denoting the size of a maximum matching of $G_f$ by $\mu(G_f)$, it is well-known that
$$
    \dmon{f} = \mu(G_f).
$$
Therefore, all we need to do is to estimate the size of maximum matching in $G_f$. While the maximum matching problem is extremely well-studied in the streaming setting, the main difference here is that we only receive the graph $G$ and function $f$ in the stream which define $G_f$ implicitly. Indeed, even for checking the existence of a single edge $(u, v)$ in $G_f$ one needs to solve the $st$-reachability problem in $G$, which itself takes $\sqrt{n}^{1+o(1)}$ passes.

\paragraph{Query complexity of maximum matching size.} Since direct access to the graph $G_f$ is expensive in our setting, we would like to estimate its maximum matching size without {\em querying} too many edges of $G_f$. Luckily, this is also a very well-studied problem. For example, an algorithm of \cite{Behnezhad21} obtains a $(2, \epsilon n)$-approximation of the maximum matching size by making $\widetilde{O}(n)$ adjacency-matrix queries to the graph. That is, each such query specifies two vertices and the response is whether they are adjacent. Each such query to $G_f$ can easily be answered in our setting by two calls to an $st$-reachability algorithm of \cite{JLS19}, hence the final algorithm will take
$$
    \underbrace{\sqrt{n}^{1+o(1)}}_{\text{cost of a pair query}} \cdot \underbrace{\widetilde{O}(n)}_{\text{number of pair queries}} = n^{3/2+o(1)}
$$
passes in total. Unfortunately, this significantly exceeds our desired $\sqrt{n}^{1+o(1)}$ pass-complexity. But it is known (see \cite{Behnezhad21,ParnasR07}) that $\Omega(n)$ adjacency-matrix queries are indeed needed to obtain any constant approximation of maximum matching size, so this black-box reduction cannot possibly result in our desired bound.\footnote{This lower bound holds even in the more powerful access model where adjacency-matrix, adjacency-list index ($i$-th neighbor of $v$), and degree queries can be mixed adaptively.}


Our general idea for breaking through the $n\sqrt{n}$ barrier discussed above is to allow the matching size estimator to make {\em stronger} queries to the graph that can be implemented just as efficiently as adjacency-matrix queries in our model.

\paragraph{The vertex-query model.} Our first observation is that the algorithm of \cite{JLS19} is actually not specific to $st$-reachability, but rather solves the more general single source reachability problem in $\sqrt{n}^{1+o(1)}$ passes. 
For us, this implies that for a vertex $v$, we can gather the entire neighbor set $N_{G_f}(v)$ of $v$  in just $\sqrt{n}^{1+o(1)}$ passes. 
This motivates the study of the query complexity of maximum matching in the {\em vertex query}\footnote{Some 
works in the literature of sublinear time algorithm also refer to such queries as {\em all-neighbor queries}.} model:

\begin{definition}[Vertex queries]
\label{def:vertex-query}
In the {\em vertex-query access model}, the algorithm queries a vertex $v \in V$, and receives its neighborhood $N(v)$.
\end{definition}

Indeed, we show that this stronger query access model allows us to break the $\Omega(n)$ query lower bound for estimating the maximum matching size in the more standard 
query models. Namely, we prove that:

\begin{graytbox}
\begin{theorem}
\label{thm:vertex-query}
    Given an $n$-vertex graph $G$, there is a randomized algorithm that $(3, \eps n)$-approximates the size of maximum matching of $G$ for any $\epsilon > 0$ using $\sqrt{n} \cdot \poly(\log n, 1/\epsilon)$ vertex queries to $G$ w.h.p.
\end{theorem}
\end{graytbox}

Plugging \cref{thm:vertex-query} in the final algorithm instead of the pair query algorithm of \cite{Behnezhad21} improves the number of passes to
$$
    \underbrace{\sqrt{n}^{1+o(1)}}_{\text{cost of a vertex query}} \cdot \underbrace{\widetilde{O}(\sqrt{n})}_{\text{number of vertex queries}} = n^{1+o(1)},
$$
which is much fewer than $n^{3/2+o(1)}$ but still far from our desired $\sqrt{n}^{1+o(1)}$ bound. 

\paragraph{The subset-query model.} Our final observation is that the algorithm of \cite{JLS19} can in fact be used to solve a {\em subset query} in just $\sqrt{n}^{1+o(1)}$ passes. Such queries are more powerful than vertex queries, and are defined as follows:

\begin{definition}[\textbf{Subset Queries}]
\label{def:subset-query}
In the {\em subset-query access model}, the algorithm queries a subset $S \subseteq V$, and receives its neighborhood
$
N(S)=\{u \in V : \exists v \in S \text{ such that } \{u,v\} \in E\}.
$
\end{definition}

For us, it is crucial that each such query only returns the set $N(S)$ and {\em not} which vertices of $S$ each vertex in $N(S)$ is adjacent to, as otherwise the algorithm of \cite{JLS19} would not be applicable.

We show that subset queries allow for a more dramatic improvement in the query-complexity: only polylogarithmically many subset queries are required to estimate the matching size. The following theorem is our main technical contribution in this work:

\begin{graytbox}
\begin{theorem}
\label{thm:mwu-subset-query-m}
    Given an $n$-vertex graph $G$, for any $\epsilon > 0$, there is a randomized algorithm that $(1.5 + \epsilon)$-approximates the size of maximum matching in $G$ using $\poly(\log n, 1/\epsilon)$ subset queries to $G$ w.h.p.
    If the graph is bipartite, then the approximation ratio improves to $1 + \epsilon$.
\end{theorem}
\end{graytbox}

Now when we plug \cref{thm:mwu-subset-query-m} in the final algorithm, the number of passes will be
$$
    \underbrace{\sqrt{n}^{1+o(1)}}_{\text{cost of a subset query}} \cdot \underbrace{\poly(\log n)}_{\text{number of subset queries}} = \sqrt{n}^{1+o(1)},
$$
which immediately implies \cref{thm:main}. We note that we apply \cref{thm:mwu-subset-query-m} in a white-box manner in our final algorithm for \cref{thm:main}.

We also give a self contained proof for $(2+\epsilon)$-approximation of maximum matching.

\begin{theorem}
\label{thm:subset-query}
    Given an $n$-vertex graph $G$, for any $\epsilon > 0$, there is a randomized algorithm that $(2+\epsilon)$-approximates the size of maximum matching in $G$ using $\poly(\log n) \cdot (1/\epsilon)^{O(1/\eps)}$ subset queries to $G$ w.h.p.
\end{theorem}

\paragraph{Approximating maximum matching size using vertex queries.}

We describe the high-level approach to proving \cref{thm:subset-query}, which shows how to approximate the size of the maximum matching in $G$ using subset queries.

Our starting point is a ``peeling''-type algorithm \cite{ParnasR07, OnakR10, KKS14,CALM18, KMN20} which iteratively matches and removes highest degree vertices. While we can't implement this algorithm in our setting, it is useful to first describe it and discuss what goes wrong in implementing it. In an idealized peeling algorithm, we divide the vertices into two halves randomly. Then each vertex from the first half ``proposes'' to a single neighbor chosen uniformly at random and independently. Then every vertex in the second half that receives at least one proposal ``accepts'' one of those arbitrarily and the two vertices get matched together. Now fix a vertex $v$ and suppose that $\deg(v) \geq \Delta/2$ where $\Delta$ is the current maximum degree of the graph. Under the event that $v$ is in the second half, it receives a proposal with probability at least $1-(1-1/2\Delta)^{\Delta/2} \geq 0.1$. 
This means that after this process, a constant fraction of vertices of degree at least $\Delta/2$ get removed from the graph. We can therefore afford to 
remove all vertices of degree at least $\Delta/2$ from the graph regardless of whether they are matched, 
since a constant fraction of them have been matched in expectation. This reduces the maximum degree from $\Delta$ down to $\Delta/2$, so $O(\log \Delta)$ iterations of this process results in a constant approximate maximum matching in expectation.

As discussed, we cannot actually carry out this process in our setting. The main reason is that each vertex needs to propose to a randomly selected neighbor. Implementing this step requires $\Theta(n)$ neighbor queries, and subset queries are not helpful in reducing the query complexity to $\poly\log(n)$. 
To get around this problem, we implement a different variant of the peeling algorithm that does not explicitly build a matching, but only provides an estimate of the matching size.

Our first ingredient is a degree estimation procedure in the subset query model, which we call a \emph{subset degree estimator}.
The subset degree estimator has the following functionality:
Given a set of active vertices $U \subseteq V$,  for every vertex $v \in V$ the procedure returns a $(1+\epsilon)$-\emph{multiplicative} approximation to $d_v := |N(v) \cap U|$. The procedure uses $O\left( \frac{\log^2 n}{\eps^3} \right)$ subset queries and succeeds with probability $1-1/\poly(n)$.
Its underlying technique is a sampling argument that uses geometrically decreasing sampling probabilities.

To see how the estimator works, fix a vertex $v$ and let $d_v = |N(v) \cap U|$, noting that the value of $d$ is unknown to the algorithm. For a sampling probability $p$, include every vertex of $U$ independently with probability $p$ to set $S$ and query $N(S)$ using a subset query.
The event that $v \notin N(S)$ happens when none of the $d$ neighbors of $v$ is sampled, which happens with probability  $(1-p)^{d_v}$.
When $p \approx 1/2d_v$, this probability is bounded away from 0 and 1 by a constant, hence by enough repetitions and a Chernoff bound, one could estimate the value of $(1-p)^{d_v}$, and thus $d_v$, multiplicatively.
By trying out different values of $p$ between $1/\poly(n)$ and 1 geometrically, we ensure that for each vertex $v$, one of these trials achieves $p \approx 1/2d_v$.  The advantage of this approach is that a single set of subset queries suffices to estimate $|N(v) \cap U|$ for all vertices $v$ simultaneously.

We next modify the standard peeling algorithm so that it constructs an implicit fractional matching rather than an integral matching. Suppose the current degree threshold is $\Delta_{\max}$, and let $U$ be the set of vertices that have not yet been removed. Let $H \subseteq U$  contain the vertices whose degree in the induced graph of $U$ is approximately $\Delta_{\max}$. For every $u \in H$, we add a small amount of fractional weight, on the order of $\epsilon/\Delta_{\max}$, to each edge from $u$ to another vertex in $U$. Therefore, the total weight of a vertex (that is the sum of weights on its edges) $u \in H$ increases proportionally to its degree in $G[U]$,  whereas the total weight of a vertex $v\in U\setminus H$ increases proportionally to the number of its neighbors in $H$. Vertices are removed from $U$ once their saturation becomes close to $1$. Since every vertex in $H$ has degree comparable to $\Delta_{\max}$, its total weight increases by $\Omega(\epsilon)$ in every iteration, and hence all vertices in $H$ are removed after $O(1/\epsilon)$ iterations. We can then decrease $\Delta_{\max}$ and repeat the process on the remaining set $U$.

The subset degree estimator allows us to implement this fractional peeling process without storing all the edges of the graph. At the beginning of an iteration, we estimate $d_U(v)=|N(v)\cap U|$ for all vertices  to identify the high degree set $H$. We then run the estimator again with $H$ as the input, obtaining estimates of
$
d_H(v)=|N(v)\cap H|
$
for all vertices. The estimates $d_U(v)$ and $d_H(v)$ determine the increase in each vertex saturation, so the algorithm only needs to maintain the active set and one saturation value per vertex.

Running this fractional process until every vertex is removed would naturally approximate the maximum fractional matching value. Using only the general $3/2$ integrality gap bound for fractional matchings gives a $(3+\epsilon)$-approximation to the maximum matching. To obtain a $(2+\epsilon)$-approximation, we use a separate low degree phase. The low degree phase entails additional technicality that is described in \cref{sec:low-deg-implementation}.

Finally, using our subset degree estimator within the MWU framework of \cite{Liu24}, we obtain a $(1+\epsilon)$-approximation to the minimum fractional vertex cover using  $\poly(\log n,1/\epsilon)$ subset queries. This gives a $(1.5+\epsilon)$-approximation to maximum matching in general graphs, and a $(1+\epsilon)$-approximation for bipartite graphs. This is proved in \cref{sec:mwu}.



\subsection{Future Directions and Open Problems} 

Although adjacency-matrix and adjacency-list queries are the standard models in the literature on sublinear-time and query algorithms, our work highlights the importance of studying two stronger models: vertex queries and subset queries. For maximum matching, we establish strong separations between these models and the traditional query models, for which $\Omega(n)$ queries are necessary. An interesting open question is whether similar separations arise for other natural problems.


\subsection{Further Related Work}\label{sec:relatedwork}

\paragraph{Monotonicity testing.} The query complexity of testing whether a function is 
monotone, namely distinguishing monotone functions from those
that are $\epsilon$-far from monotone (in the Hamming distance), 
is one of the earliest property testing problems to be considered \cite{ErgunKKRV00,GoldreichGLRS00}.  
The tolerant version of the problem, 
distinguishing functions that are $\eps$-close to monotone, 
from those that are $\eps'$-far from monotone is 
considered in the work that  introduces tolerant testing \cite{ParnasRR06} 
(see also \cite{AilonCCL07}),    where it was demonstrated  that 
tolerant testing  and the estimating the distance to having the 
property are essentially equivalent tasks. 
Since then, a large body of work  has considered algorithms
and lower bounds for testing  and tolerant testing of
monotonicity over various domains and distance measures
(examples include \cite{DodisGLRRS99,FischerLNRRS02,
Fischer04,
AilonCCL08,
SaksS10,
BhattacharyyaGJJRW10,
BrietCGM12,
BermanRY14,
ChakrabartyS13a,
ChakrabartyS14,
ChenST14,
ChenDST15,
ChakrabartyS16,
ChenWX17,
KhotMS18,
BlackCS18,
ChakrabartyS19,
BelovsB15,
BelovsB21,
LangeRV22,
BlackCS20,
PallavoorRW22,
BlackC023,
BravermanKKM23,
Pinto23,
Pinto24,
ChenDLNS24,
BlackKR24,
Chakrabarty025,
ChenDHLNSY25,
FeiP25,Yoshida26}).

\paragraph{Property testing vs. streaming algorithms.}
A connection between
property testing and streaming has been explored in \cite{MonemizadehMPS17,PengS18,CzumajF0S20}
where it is shown how to transform constant query property testers
to the random-order single-pass streaming model with constant space.
Note that the aforementioned lower bounds demonstrate that in most settings
that have been considered, monotonicity testing 
requires nonconstant dependence on the domain size of the function.

\paragraph{Sublinear time algorithms for maximum matching.} Sublinear-time matching algorithms have been a central topic of study and have enjoyed sustained interest in the past few decades
\cite{ParnasR07,NguyenO08,YoshidaYI09,OnakRRR12,LeviRY17,KapralovMNT20,Behnezhad21,BehnezhadRRS-SODA23, BehnezhadRR23STOC,BhattacharyaKS23STOC,BehnezhadRR23FOCS,BehnezhadRR24,AzarmehrBRR25,MahabadiRT25}.
Beyond addressing a fundamental problem in graph theory, they are deeply connected to many other settings and problems, including spanning trees, Steiner trees, and the traveling salesman problem \cite{ChenKK20,ChenKT23ICALP,ChenKT23SODA,BehnezhadRRS24,MahabadiRTV25ITCS,MahabadiRTV26SODA}.
In addition, they have served as a key tool in dynamic algorithms for estimating the maximum matching size \cite{Behnezhad23,BhattacharyaKS23FOCS,BhattacharyaKSW24betterthantwo,AzarmehrBR24},
and more recently in \cite{BehnezhadG24,AssadiKK25} for maintaining an approximate maximum matching in dynamic graphs, with an update time tied to the density of Ruzsa-Szemer\'edi graphs. This work serves as  yet another important implication of such algorithms.


\section{Preliminaries}

\subsection{Definitions and Notation}
\label{sec:definitions}
For a graph $G = (V,E)$, $n$ will refer to $|V|$, and we will denote by $\mu(G)$ the size of a maximum matching in $G$.
We say that an algorithm succeeds with high probability if for any constant $c > 0$, the algorithm can be made to succeed with probability at least $1 - n^{-c}$. 
We call an estimate $\hat{x}$ of a value $x$ an $\alpha$-approximation if it satisfies
\[x/\alpha \le \hat{x} \le x,\]
and an $(\alpha,\beta)$-approximation if it satisfies
\[x/\alpha - \beta \le \hat{x} \le x.\]
\begin{definition}[Edge and label stream]
    \label{def:edge-label-stream}
Let $G = (V, E)$ be a directed acyclic graph and $f: V \to \{0,1\}$ be a boolean function. 
In an edge and label stream of $G$ and $f$, edges of $G$ and example-label pairs $(v, f(v))$ 
are revealed in an arbitrary order.
\end{definition}

\begin{definition}[Violation graph]
For a DAG $G=(V,E)$ and a boolean function $f : V \to \{0,1\}$, the violation graph $G_f$ is the undirected graph which has an edge $\{u, v\}$ if $u$ can reach $v$ but $f(u) > f(v)$ (or $v$ can reach $u$ but $f(v) > f(u)$).
\end{definition}


\subsection{Standard Facts}

\begin{fact}[Distance to monotonicity and matching in $G_f$]
\label{fact:matching-dmon}
For a DAG $G=(V,E)$ and a boolean function $f : V \to \{0,1\}$, we have
$\dmon{f} = \mu(G_f).$
\end{fact}

\begin{fact}[Integrality gap bound for fractional matchings]
\label{fact:frac-gap-three-halves}
For every graph $G$,
$
\nu_f(G) \le \frac32 \mu(G),
$
where $\nu_f(G)$ denotes the maximum size of a fractional matching in $G$.
\end{fact}

%

%


\section{Global Matching Estimation}\label{sec:main-algorithm}
In this section, we prove validity and approximation ratio of the matching value computed by a ``global'' algorithm, \textsc{EstimateMatchingSize}
(\Cref{alg:matching}).
This algorithm accesses the graph $G$ via black-box oracles \textsc{ComputeDegrees}
and \textsc{MaximalMatching}
which provide the following functionality:
\begin{assumption}[\textsc{ComputeDegrees}]
\label{assume:compute-degrees}
There is a function $\textsc{ComputeDegrees}(G,U, \eps)$ that takes an arbitrary graph $G=(V,E)$, 
subset $U \subseteq V$, and approximation parameter $\eps > 0$, and returns a vector $\vec{d}$ such that $d_v = (1 \pm \eps)|N(v) \cap U|$ for all $v \in V$.
\end{assumption}

\begin{assumption}[\textsc{MaximalMatching}]
\label{assume:maximal-matching}
There is a function $\textsc{MaximalMatching}(G,U)$ that takes an arbitrary graph $G=(V,E)$ and 
subset $U \subseteq V$, and returns a maximal matching in $G[U]$.
\end{assumption}
Only these functions interact with the edges of $G$ (and thus require an access model to be implemented);
the rest of the algorithm only maintains 
information about the active set $U$ and the saturation vector $\vec{x}$.
In \Cref{sec:subset-query-implementation}, we will implement these oracles in
the subset-query model,
and in \Cref{sec:vertex-query-implementation} we will discuss how to adapt these ideas to the vertex-query model.

\begin{algorithm}
\DontPrintSemicolon
\caption{$\textsc{EstimateMatchingSize}(G, \eps)$}
\label{alg:matching}
\KwIn{Graph $G = (V,E)$, approximation parameter $\eps > 0$}
\KwOut{Constant-factor estimate of the maximum matching size in $G$}
\BlankLine
$U \gets V$ \tcp*{Set of active vertices; i.e. those that haven't been deleted}
$\dmax \gets n$ \tcp*{degree threshold}
$\eta \gets 2/\eps^2$ \tcp*{Transition threshold to low-degree phase}
$\vec{x} := (x_v)_{v \in V} \gets \vec{0}$ \tcp*{vector of fractional matching weights assigned to vertices}
\tcc{$\vec{y} := (y_e)_{e \in E} \gets \vec{0}$; this is a fractional matching we track in the comments to
help with the analysis, but it is never explicitly maintained by the algorithm}
\While{$\dmax \ge \eta$}{ 
    $U, \vec{x} \gets \textsc{PeelHighDegree}(G, U, \vec{x}, \dmax, \eps)$ \\
    $\dmax \gets (1 - \eps)\dmax$
}
$\vec{x} \gets \textsc{MatchLowDegree}(G,U,\vec{x}, \eps)$\\
\Return $\hat{\mu} := \tfrac{1-\eps}{1+\eps} \cdot \tfrac12\sum_{v} x_v$ 
\end{algorithm}

\begin{algorithm}
\DontPrintSemicolon
\caption{$\textsc{PeelHighDegree}(G, U, \vec{x}, \dmax, \eps)$}
\label{alg:peel-high-degree}
\KwIn{Graph $G = (V,E)$, active set $U$, saturation values $\vec{x}$, degree threshold $\dmax$, approximation parameter $\eps > 0$}
\KwOut{Updated active set and saturation values}
\BlankLine
$\vec{d_U} \gets \textsc{ComputeDegrees}(G, U, \eps)$\tcp*{counts neighbors in $U$, see \Cref{assume:compute-degrees}}
$H \gets \{v \in U : (d_U)_v \ge (1 - \eps)\dmax\}$\\
\While{$H \ne \emptyset$}{
    $\vec{d_H} \gets \textsc{ComputeDegrees}(G, H, \eps)$\\
    $x_v \gets x_v + \frac{\eps}{2} \cdot \frac{(d_U)_v}{\dmax}$ for all $v \in U$\\
    $x_v \gets x_v + \frac{\eps}{2} \cdot \frac{(d_H)_v}{\dmax}$ for all $v \in U \setminus H$   \\
    \tcc{$y_{uv} \gets y_{uv} + \frac{\eps}{2\dmax}$ for each ordered pair $(u,v) \in E(G[U])$ with $u \in H$}
    $H \gets H \setminus \{v \in H : x_v \ge 1 - 2\eps\}$ \\
    $U \gets U \setminus \{v \in U : x_v \ge 1 - 2\eps\}$ \\
}
\Return $U, \vec{x}$\\
\end{algorithm}

\begin{algorithm}
\DontPrintSemicolon
\caption{$\textsc{MatchLowDegree}(G,U,\vec{x}, \eps)$}
\label{alg:low-degree-matching}
\KwIn{Graph $G = (V,E)$, active set $U$, saturation vector $\vec{x}$, approximation parameter $\eps > 0$}
\KwOut{Updated saturation vector}
\BlankLine
\For{$1/\eps - 1$ iterations}{
    $M \gets \textsc{MaximalMatching}(G,U)$\tcp*{see \Cref{assume:maximal-matching}}
    $x_v \gets x_v + \eps$ for all $v \in M$\\
    \tcc{$y_{uv} \gets y_{uv} + \eps$ for each $\{u,v\} \in M$}
    $U \gets U \setminus \{v \in U : x_v \ge 1 - 2\eps\}$\\
}
\Return $\vec{x}$
\end{algorithm}

\subsection{Validity}
Here we will show that the value returned by \textsc{EstimateMatchingSize} is at most $\mu(G)$.
To do this, we will first show that the vertex weights $\vec{x}$ correspond to an implicit fractional matching 
$\vec{y}$ that is always valid.
Then, we will show that this fractional matching has an integrality gap of at most $1/(1-\eps)$, which we 
account for by scaling down all the $x_v$ values by $1-\eps$ at the end of the algorithm.

\begin{lemma}[Invariants of \textsc{PeelHighDegree}]
\label{lem:high-degree-validity}
Suppose \textsc{PeelHighDegree} (\Cref{alg:peel-high-degree}) is called with parameters satisfying 
\begin{align*}
\epsilon &\le \tfrac13 \qquad \text{and} \qquad
\dmax \ge (1-\eps) \max_{v \in G} |N(v) \cap U|
\end{align*}
and such that the following conditions are satisfied:
\begin{enumerate}
    \item For every $v \in V$, $(1-\eps) \sum_{u \in N(v)} y_{uv} \le x_v \le (1+\eps) \sum_{u \in N(v)} y_{uv}$.
    \item For every $v \in V$, $x_v \le 1 - \eps$. 
    \item For every $v \in U$, $x_v \le 1-2\eps$. 
\end{enumerate}
Then these three conditions are maintained by \textsc{PeelHighDegree}.

Furthermore, when \textsc{PeelHighDegree} returns,
all vertices in $U$ have degree at most $\dmax$.
\end{lemma}
\begin{proof}
First, we observe that all vertices of degree at least $\dmax$ are in $H$; this is because \Cref{assume:compute-degrees} ensures 
that all such vertices have $d_U(v) \ge (1-\eps) \dmax$ and are thus placed in $H$.
Since all vertices in $H$ are removed from $U$ during the execution of \textsc{PeelHighDegree}, this ensures that all vertices remaining in $U$ afterwards have degree at most $\dmax$.

Now we show the invariants.
Assume the invariants hold at the beginning of the $t$-th execution of the loop body.
For $v \in V \setminus U$, $x_v$ and all $y_{uv}$ incident to $v$ are unchanged, so we will analyze only vertices that are in $U$
at the start of the $t$-th execution.
To show the first invariant, consider the following cases.
\begin{enumerate}
  \item $v \in H$: We increase $x_v$ by $\frac{\eps}{2} \cdot \frac{(d_U)_v}{\dmax}$ and $y_{uv}$ by 
  $\frac{\eps}{2\dmax}$ for all $u \in N(v) \cap U$. By \Cref{assume:compute-degrees}, 
  we have $(d_U)_v = (1 \pm \eps) |N(v) \cap U|$. Thus, after the update to $x_v$, we have 
  \begin{align*}
  (x_v)_{t+1}
  &= (x_v)_t + \frac{\eps}{2} \cdot \frac{(d_U)_v}{\dmax}\\
  &\ge \left[(1 - \eps) \sum_{u \in N(v) \cap U} (y_{uv})_t\right] + \frac{\eps}{2} \cdot \frac{(1-\eps)|N(v) \cap U|}{\dmax}\\
  &= (1 - \eps) \sum_{u \in N(v) \cap U} \left[(y_{uv})_t + \tfrac{\eps}{2\dmax}\right] \\
  &= (1 - \eps) \sum_{u \in N(v) \cap U} (y_{uv})_{t+1}. 
  \end{align*}
  A similar argument shows the upper bound of $(1+\eps) \sum_{u \in N(v)} y_{uv}$.

  \item $v \in U \setminus H$: We increase $x_v$ by $\frac{\eps}{2} \cdot \frac{(d_H)_v}{\dmax}$ and $y_{uv}$ 
  by $\frac{\eps}{2} \cdot \frac{1}{\dmax}$ for all $u \in N(v) \cap H$. We have:
  \begin{align*}
  (x_v)_{t+1}
  &= (x_v)_t + \frac{\eps}{2} \cdot \frac{(d_H)_v}{\dmax}\\
  &\ge \left[(1 - \eps) \sum_{u \in N(v) \cap U} (y_{uv})_t\right] + \frac{\eps}{2} \cdot \frac{(1-\eps)|N(v) \cap H|}{\dmax}\\
  &= \left[(1 - \eps) \sum_{u \in N(v) \cap (U\setminus H)} (y_{uv})_t\right] + 
  \left[(1 - \eps) \sum_{u \in N(v) \cap H} (y_{uv})_t + \tfrac{\eps}{2\dmax}\right] \\
  &= (1 - \eps) \cdot  \sum_{u \in N(v) \cap U} (y_{uv})_{t+1}. 
  \end{align*} 
  A similar argument shows the upper bound of $(1+\eps) \sum_{u \in N(v)} y_{uv}$.

\end{enumerate}

To show the other invariants, first observe that the third invariant follows from the fact that the vertices not satisfying
$x_v \le 1-2\eps$ are removed from $U$ at the end of the loop body.
For the second, observe that $x_v$ for $v \in U$ increases by $\frac{\eps}{2} \cdot \frac{(d_U)_v}{\dmax}$ 
or $\frac{\eps}{2} \cdot \frac{(d_H)_v}{\dmax}$ in each iteration of the loop. By \Cref{assume:compute-degrees} and the assumption that 
$\dmax \ge (1 - \eps) \max_{v \in G} |N(v) \cap U|$, we have 
\[\tfrac{(d_U)_v}{\dmax} \le \tfrac{1+\eps}{1-\eps} \le 2,\]
and similarly for $\tfrac{(d_H)_v}{\dmax}$.
Thus the value of $x_v$ increases by at most $\eps$ in each iteration. Since $(x_v)_t \le 1 - 2\eps$ at the start of the iteration,
$(x_v)_{t+1} \le 1 - \eps$. 
\end{proof}

We now state and prove a similar statement for \textsc{MatchLowDegree}.
\begin{lemma}[Invariants of \textsc{MatchLowDegree}]
\label{lem:low-degree-invariants}
Suppose \textsc{MatchLowDegree} (\Cref{alg:low-degree-matching}) is called under the conditions: 
\begin{enumerate}
    \item For every $v \in V$, $(1-\eps) \sum_{u \in N(v)} y_{uv} \le x_v \le (1+\eps) \sum_{u \in N(v)} y_{uv}$.
    \item For every $v \in V$, $x_v \le 1 - \eps$.
    \item For every $v \in U$, $x_v \le 1-2\eps$.
\end{enumerate}
Then these three conditions are maintained by \textsc{MatchLowDegree}.
\end{lemma}

\begin{proof}
Assume the invariants hold at the start of iteration $t$. As before, the third invariant follows
from the fact that vertices not satisfying $x_v \le 1-2\eps$ are removed from $U$ at the end of the iteration,
and the second invariant follows from the third invariant and the fact that $x_v$ increases by at most $\eps$ in each iteration.

To show the first invariant, consider a vertex $v$ that is matched in $M$ during the $t$-th iteration, since
otherwise $x_v$ and any $y_{uv}$ incident to $v$ would be unchanged. 
Let $w$ be $v$'s matching partner and note that all $y_{uv}$ for $u \ne w$ are also unchanged. We have then:
\begin{align*}
(x_v)_{t+1}
= (x_v)_t + \eps
&\ge \left[ (1 - \eps) \sum_{u \in N(v)} (y_{uv})_t \right] + \eps \\
&\ge \left[ (1 - \eps) \sum_{u \in N(v), u \ne w} (y_{uv})_t \right]+ (1-\epsilon)\left[(y_{wv})_t + \eps \right]\\ 
&= (1 - \eps) \sum_{u \in N(v)} (y_{uv})_{t+1}.
\end{align*}
A similar argument shows the upper bound of $(1+\eps) \sum_{u \in N(v)} y_{uv}$.
\end{proof}

Now we move on to bound the integrality gap of the fractional matching $\vec{y}$.
We will use the following well-known fact, which gives a sufficient condition for a fractional matching 
to be close to an integral matching:

\begin{fact}[Small-set blossom inequalities]
\label{lem:blossom-rounding}
Suppose $y : E \to [0, 1]$ is a fractional matching such that for every odd-size set $S \subset V$ with $|S| < \frac{1}{\beta}$, the \emph{blossom inequality}
\[
\sum_{e \in E(G[S])} y_e \le \frac{|S| - 1}{2}
\]
holds.
Then there exists an integral matching $M$ such that $|M| \ge (1 - \beta)\sum_{e \in E} y_e$.
\end{fact}
This is because the fractional matching $(1 - \beta)y$ satisfies the blossom inequalities for all odd-size sets, which implies it is a convex combination of integral matchings \cite{E65}.
We will show that our fractional matching satisfies this condition with $\beta = \eps$:

\begin{lemma}
  \label{lem:fractional-validity}
Assume $\eps \le \tfrac13$, and  let $\vec{x}$ be the vector computed by \textsc{EstimateMatchingSize} (\Cref{alg:matching}). 
Then there exists a fractional matching in $G$ with vertex saturation vector $\frac{1}{1+\eps} \vec{x}$
that satisfies the blossom inequalities for all odd sets $S$ of size at most $1/\eps$.
\end{lemma}
\begin{proof}
It follows from \Cref{lem:high-degree-validity,lem:low-degree-invariants} that the invariants on $\vec{x}$ and $\vec{y}$ are preserved throughout the execution of \textsc{EstimateMatchingSize}.
Therefore the final values of $\vec{x}$ and $\vec{y}$ satisfy $\frac{1}{1+\eps} x_v \le \sum_{u \in N(v)} y_{uv} \le 1$ for every $v$;
thus $\vec{y}$ is a fractional matching with vertex saturations at least $\frac{1}{1+\eps} \vec{x}$.

Now we claim the $\vec{y}$ values satisfy the small-set blossom inequalities.

Let $L$ denote the set of vertices in $U$ when \textsc{MatchLowDegree} is called.
First we will bound the weight on edges with an endpoint outside of $L$.
Observe that in \textsc{PeelHighDegree}, since $\frac{(d_U)_v}{\dmax} \ge 1-\eps \ge 1/2$, the value of $x_v$ for $v \in H$ increases by
at least $\eps/4$ each iteration;
thus the number of iterations is at most $4/\eps$.
Each edge incident to $H$ gains weight $\tfrac{\eps}{2\dmax}$ in each iteration,
so after $4/\eps$ iterations, each edge will have gained weight at most $2/\dmax \le 2/\eta = \eps^2$. 
Only edges incident to $H$ gain weight, and each vertex appears in $H$ at most once; thus, all edges with an endpoint
outside of $L$ satisfy $y_{uv} \le \eps^2$.

Consider a set $S \subseteq V$ of odd size $k \in [3, 1/\eps]$, with $\ell$ vertices in $L$ and $k - \ell$ vertices outside of $L$.
We will bound the matching weight in $G[S]$.
The contribution from edges with an endpoint outside of $L$ is at most $\eps^2 \cdot k(k-\ell) \le \eps(k-\ell)$ by the above argument.  
For the edges between members of $L$, all their weight comes from \textsc{MatchLowDegree}. Since there are $1/\eps - 1$ iterations
and each contributes $\eps$ weight to an integral matching, the contribution from $L$ is at most $(1-\eps)\left\lfloor\frac{\ell}{2} \right\rfloor$; thus the inequality is immediately satisfied if $\ell = k$. 

If $\ell \le k-1$, then the total weight is at most 
\begin{align*}
(1-\eps) \frac{\ell}{2} + \eps(k-\ell)
&\le (1 - \eps) \frac{k - 1}{2} + \eps \cdot 1 \\
&\le \frac{k - 1}{2},
\end{align*}
where the first inequality is valid because $\frac{1 - \eps}{2} \ge \eps$, and the second is because $k \ge 3$.
\end{proof}

\begin{corollary}[Validity of matching]
  \label{cor:validity}
Let $\vec{x}$ be the saturation vector returned by \Cref{alg:matching} with $\eps \le 1/3$ and $\eta \ge 2/\eps^2$.
Then there exists an integral matching in $G$ of value at least $\tfrac{1-\eps}{1+\eps} \cdot \tfrac12 \sum x_v$.
\end{corollary}
\begin{proof}
By \Cref{lem:fractional-validity}, there is a fractional matching satisying the small-set blossom inequalities with
vertex saturations at least $\tfrac{1}{1+\eps} \vec{x}$. 
This matching has value at least $\tfrac{1}{1+\eps} \cdot \tfrac12\sum x_v$. 
By \Cref{lem:blossom-rounding}, 
there is thus an integral matching in $G$ of value at least $\tfrac{1-\eps}{1+\eps} \cdot \tfrac12 \sum x_v$. 
\end{proof}

\subsection{Approximation Ratio}
We have shown that the value returned by \Cref{alg:matching} is at most $\mu(G)$;
i.e. it does not overestimate the matching size.
In this subsection, we will show that \textsc{EstimateMatchingSize} achieves an approximation ratio of $2 + O(\eps)$.

\begin{lemma}
\label{lem:apx-factor}
The fractional matching value $\hat{\mu}$ returned by \Cref{alg:matching} satisfies 
\[\hat{\mu} \ge \frac{\mu(G)}{2 + O(\eps)}.\]
\end{lemma}
To prove this, we will make use of the following duality relationship between vertex covers and matchings:

\begin{fact}
\label{clm:apx-factor}
Let $x_v$ be a real vector over the vertices of a graph $G$. 
If $x_u + x_v \ge c$ for every edge $\{u,v\} \in E$, then $\sum_{v \in V} x_v \ge c\mu(G)$.
\end{fact}

\begin{proof} 
Consider a maximum integral matching $M^\star$ in $G$ and let $V(M^\star)$ denote the set 
of its endpoints.
Then
\begin{align*}
\sum_{v \in V} x_v
&\ge \sum_{\{u,v\} \in M^\star} x_u + x_v
\ge c\mu(G). \qedhere
\end{align*}
\end{proof}
Using this we can proceed with the bound on the approximation ratio of our algorithm.
\begin{proof}[Proof of \Cref{lem:apx-factor}]
We will prove that every edge $(u,v) \in G$ satisfies $x_u + x_v \ge 1 - 2\eps$.
Let $L$ be the set of vertices in $U$ when \textsc{MatchLowDegree} is called.
Consider the following 2 cases: 
\begin{enumerate}
\item $u \in V \setminus L$ or $v \in V \setminus L$: In this case, one of $u$ or $v$ is removed by the subroutine
$\textsc{PeelHighDegree}$ (\Cref{alg:peel-high-degree}). Since vertices are only removed under the condition 
$x_v \ge 1 - 2\eps$, we have $x_u + x_v \ge 1 - 2\eps$. 
\item $u, v \in L$: Either $u$ or $v$ is removed due to its value exceeding $1-2\eps$, or the edge $(u,v)$ persists
through the entirety of \textsc{MatchLowDegree}. In the latter case, since every matching $M$ is maximal, 
either $u$ or $v$ must be matched in $M$. Thus in each iteration, $\eps$ is added to the value of $x_u$ or $x_v$.
Since there are $1/\eps - 1$ iterations, we have $x_u + x_v \ge 1-\eps$.
\end{enumerate}
We have shown that in both cases, $x_u + x_v \ge 1 - 2\eps$; thus \Cref{clm:apx-factor} shows that $\sum_{x \in V} x_v \ge (1 - 2\eps)\mu(G)$, which concludes the proof. 
\end{proof}

\section{Implementation With Subset Queries}
\label{sec:subset-query-implementation}
We have given an algorithm that estimates matching size, given an oracle that 
provides degree estimates and an oracle that provides maximal (integral) matchings.
In this section, we will modify this algorithm to work with subset-query access to the graph (\Cref{def:subset-query}).
In \Cref{sec:subset-degree-est} will straightforwardly implement the \textsc{ComputeDegrees} oracle with subset queries.

Then, in \Cref{sec:low-deg-implementation}, we will present a query-efficient alternative to \Cref{alg:low-degree-matching}: instead of computing global matchings, we give a \emph{local computation algorithm (LCA)} that computes the output value $x_v$ at a single vertex $v$ by exploring (in expectation) $\Delta^{O(1/\epsilon)}$ vertices near $v$.
The algorithm works by recursively running the LCA for randomized greedy maximal matching, with each level of recursion representing an iteration of \cref{alg:low-degree-matching}.
This allows us to determine how many times $v$ is matched and thus how much to increase $x_v$.
We use this to sample from the output vector $\vec{x}$ in order to estimate the final matching value.

\subsection{Degree Estimation}
\label{sec:subset-degree-est}
In this section, we will implement the functionality of \textsc{ComputeDegrees} as 
specified in \Cref{assume:compute-degrees}, in the subset-query model. The goal is to
prove the following:
\begin{lemma}[Estimating degrees with subset queries] \label{lem:degree-estimation-subset-queries}
There is an algorithm that uses \newline $O\left( \frac{\log^2 n}{\eps^3} \right)$ subset queries and $O(n \log (n/\eps))$ space, implements the functionality of \Cref{assume:compute-degrees} when it succeeds, and succeeds with high probability.
\end{lemma}
To simplify the proof, we abstract away the details concerning the graph as follows.
The goal is to estimate a hidden value $d$ representing the (induced) degree of a vertex; i.e. the number of neighbors it has in some subset.
We are allowed to learn about $d$ by sampling independently from the distribution family $\left\{\Bernoulli_{(1-p)^d}\right\}_{p \in [0, 1]}$.
This is equivalent to sampling every vertex with a probability $p$, and checking if $v$ is a neighbor of any of them, which happens with probability $1 - (1 - p)^d$.
We prove that such Bernoulli samples, represented by $Y_i^{(r)}$ in \cref{alg:geo-sampling}, can be used to estimate the value of $d$ up to a multiplicative factor of $(1 + \epsilon)$.

\subsubsection{Degree Estimation Through Geometric Sampling}

%
\begin{algorithm}
    \caption{Degree Estimation Through Geometric Sampling}
    \label{alg:geo-sampling}
    \KwIn{
    Domain size $n$, accuracy parameters $\epsilon$ and $\delta$,  access to independent samples from distributions $\mc{D}_p := \Bernoulli((1 - p)^d)$ for $p \in [0, 1]$, where $d \leq n$ is the hidden value}
    \KwOut{
    An estimate for $d$
    }
    \BlankLine
    Let $L \gets \ceil{\log_{1 + \epsilon} 4n }$, 
    $R \gets 500 \frac{\log L + \log(1/\delta)}{\epsilon^2}$,
    and $p_i\gets (1 + \epsilon)^{-i}$

    \For{$i \in [L]$ \textnormal{and} $r \in [R]$}{

        Draw $Y_i^{(r)} \sim \mc{D}_{p_i}$ independently

        $\displaystyle Z_i \gets \frac{1}{R} \sum_{r  = 1}^R Y_i^{(r)}$
    }

    $i^* \gets \min\{i \mid Z_i \geq 1/2\}$

    \Return{$\displaystyle \dh := \frac{\ln Z_{i^*}}{\ln (1 - p_{i^*})}$}
\end{algorithm}

\begin{lemma} \label{lem:geo-sampling}
    The output of \cref{alg:geo-sampling} satisfies $\dh = (1 \pm \epsilon) d$, with probability $1 - \delta$.
\end{lemma}

To provide some intuition, we note that since each $Y_i^{(r)}$ is a Bernoulli random variable with mean $q_i := (1 - p_i)^d$, $Z_i$ provides an estimate of $q_i$.
Furthermore, for any $i \in [L]$, it holds that $\frac{\ln q_i}{\ln(1 - p_i)} = d$.
The choice of $i^*$ is such that $q_{i^*} = \Theta(1)$ and equivalently $p_{i^*} = \Theta(\frac{1}{d})$.
As a result, 
$\frac{\ln Z_{i^*}}{\ln (1 - p_{i^*})}$ provides an accurate estimate of $d$.
We formalize the proof below. 
The analysis uses an additional parameter, $\lambda = \epsilon/20$.

\begin{claim} \label{clm:geo-sampling-hoeffding}
    It holds that $\card{Z_i - q_i} \leq \lambda$ for all $i$, with probability $1 - \delta$.
\end{claim}
\begin{proof}
    This follows from a direct application of Hoeffding's inequality.
    $Z_i$ is the average of $R$ independent $0$-$1$ variables, with an expected value of $q_i$.
    Therefore, for each $i$ it holds:
    $$
    \Pr(\card{Z_i - q_i} > \lambda) \leq 2e^{-2R\lambda^2} \leq \frac{\delta}{L}.
    $$
    Hence, taking the union bound over $i \in [L]$ implies that with probability $1 - \delta$, it holds that $\card{Z_i - q_i} \leq \lambda$ for all $i$.
\end{proof}

\begin{claim}
    Assuming the event in \cref{clm:geo-sampling-hoeffding} holds, we have that $q_{i^*} \in [1/4, 3/4]$.
\end{claim}
\begin{proof}    
    First, note that $q_{i^*} \geq \frac{1}{4}$ holds,
    since 
    $$q_{i^*} \geq Z_{i^*} -\lambda\geq \frac{1}{2} - \lambda \geq \frac{1}{4}.$$
    To upper-bound $q_{i^*}$, let $j = \min\{j \mid p_j \leq \frac{1}{3d}\}$.
    That is, $\frac{1}{3(1 + \epsilon)}\leq d\cdot p_j \leq \frac{1}{3}$.
    As a result, we have
    $$
    q_j = (1 - p_j)^d \geq 1 - d\cdot p_j \geq \frac{2}{3}.
    $$
    Additionally, we have that $Z_j \geq \frac{2}{3}- \lambda \geq \frac{1}{2}$, and hence $i^* \leq j$.
    Therefore, it holds
    $$
    q_{i^*} \leq q_j = (1 - p_j)^d \leq e^{-d \cdot p_j} \leq e^{-1/3(1+\epsilon)} \leq \frac{3}{4},
    $$
    which concludes the proof.
\end{proof}

\begin{claim}
    It holds that $\left\lvert\frac{\ln Z_{i^*}}{\ln(1 - p_{i^*})} - d\right\lvert \leq \epsilon d$, with probability $1 - \delta$.
\end{claim}
\begin{proof}
    Recall that $\frac{\ln q_{i^*}}{\ln (1 - p_{i^*})} = d$,
    and with probability $1 - \delta$,
    we have $\card{Z_{i^*}- q_{i^*}} \leq \lambda$ and $q_{i^*} \in [1/4, 3/4]$.
    Together, these imply
    \begin{align}
        \left\lvert\frac{\ln Z_{i^*}}{\ln(1 - p_{i^*})} - d\right\lvert 
        &= \left\lvert\frac{\ln Z_{i^*} - \ln q_{i^*}}{\ln(1 - p_{i^*})}\right\lvert  \notag\\
        &\leq \frac{4\lambda}{\card{\ln(1 - p_{i*})}} \label{eq:geo-sampling-eq1}\\
        &= \frac{ 4\lambda d }{ \card{\ln(q_i)}} \notag\\
        &\leq 16\lambda d \label{eq:geo-sampling-eq2} \\
        &\leq \epsilon d \notag.
    \end{align}
    Here, the two equalities follow from $\frac{\ln q_{i^*}}{\ln (1 - p_{i^*})} = d$, \eqref{eq:geo-sampling-eq1} follows from $Z_{i^*},  q_{i^*} \geq \frac{1}{4}$ and $\card{Z_{i^*} - q_{i^*}} \leq \lambda$, and \eqref{eq:geo-sampling-eq2} follows from $q_{i^*} \leq \frac{3}{4}$.
    This concludes the proof of the claim and \cref{lem:geo-sampling}.
\end{proof}
\subsubsection{Proof of \Cref{lem:degree-estimation-subset-queries}}
We now implement geometric sampling with subset queries, with the degree vector as the hidden values,
and we bound the query and space complexity. 
\begin{proof}[Proof of \Cref{lem:degree-estimation-subset-queries}]
First we will bound the space requirements of \Cref{alg:geo-sampling}, which we will call with $\delta := n^{-{(c+1)}}$ for some constant $c$.
The counters $i, r,$ and $Z_i$ can all be stored in $O(\log \log (n/\delta) + \log(1/\eps))$ space,
while $\hat{d}$ itself can require up to $\log((1+\eps)n) = O(\log n)$ space. 
Thus we can bound the space by $O(\log (n/\eps))$.

Observe that if \Cref{alg:geo-sampling} is run for each vertex $v$ with hidden value
$|N(v) \cap U|$ and failure probability $\delta := n^{-{(c+1)}}$, then it implements the behavior specified in \Cref{assume:compute-degrees} with probability $\ge 1 - n^{-c}$ by a union bound over vertices.
We will implement the sampling of $\Bernoulli((1 - p)^d)$ in parallel for all vertices
by sampling each member of $U$ independently with probability $p$.
Let $S$ be the set of sampled vertices and let $T$ be the set returned by the subset query. 
We have for each $v$:
\[\Pr[v \not\in T] = \prod_{u \in N(v) \cap U} u \not\in S;\]
thus the vector $(\one[v \not\in T])_{v \in V}$ has marginal distributions 
$\Bernoulli((1 - p)^{|N(v) \cap U|})$.

We run \Cref{alg:geo-sampling} in parallel for all vertices,
so the space requirement becomes $O(n\log (n/\eps))$.
The number of samples required for each vertex is $LR = O\left(\frac{\log^2 n}{\eps^2\log(1+\eps)} \right) = O\left( \frac{\log^2 n}{\eps^3} \right)$,
so this is the number of subset queries required.
\end{proof}

\newcommand{\EdgeValue}{\textsc{EdgeValue}}
\newcommand{\VertexValue}{\textsc{VertexValue}}

\subsection{Low-Degree Phase Implementation}
\label{sec:low-deg-implementation}

This section is devoted to the implementation of the low-degree integral matching phase (\cref{alg:low-degree-matching}).  Our implementation of this phase happens to use only vertex queries (which are a special case of subset queries), but we do not make use of this fact. We prove the following lemma.

\begin{lemma} \label{lem:low-deg-implementation}
Given a graph $G$ with maximum degree $\Delta$, active set $U \subseteq V(G)$, saturation vector $\vec{x}$, and an approximation parameter $\epsilon > 0$, let $v$ be a vertex chosen uniformly at random.
There is an algorithm that computes the update saturation value of $v$ as in \cref{alg:low-degree-matching} using $\Delta^{O(1/\epsilon)}$ vertex queries in expectation,
where the expectation is over the randomness of the algorithm and the choice of $v$.
\end{lemma}

To implement the low-degree phase, we employ the scheme of \cite{YoshidaYI09}.
Before describing the implementation, we review the definition of randomized greedy maximal matching (RGMM), which is used as a subroutine.
Given a graph $G$, a maximal matching can be computed by iterating the edges in a random order $\pi$, and adding edges to the matching greedily (i.e., any iterated edge that has no adjacent edge in the matching is added to the matching).
We refer to this matching as the greedy maximal matching corresponding to order $\pi$, $\textsc{GMM}(G, \pi)$, and the output of the entire process as the randomized greedy maximal matching.

\cite{YoshidaYI09} analyze a recursive implementation of this algorithm that, given a specific edge $e$ and an ordering $\pi$, determines whether $e$ appears in $\textsc{GMM}(G, \pi)$.
To do so, the algorithm recursively checks whether any adjacent edge $e'$ that appears earlier in the order is in $\textsc{GMM}(G, \pi)$.
If so, then $e$ does not. Otherwise, $e$ must be added to the maximal matching.
The algorithm is formalized below (\cref{alg:low-deg-gmm}) where, instead of a permutation, $\pi$ is a mapping from the edges to $[0, 1]$.
The key idea is that the neighboring edges can be checked in order of $\pi$, and as soon as any of them is added to the maximal matching, the rest need not be checked.
Using this improvement, they show that the following holds.

\begin{lemma}[{follows from \cite[Lemma 2.3]{YoshidaYI09}}]
    \label{lem:yyi}
    Let $Q(G, \pi, e)$, the number of \textbf{in-queries} of $e$, be defined as the number of (recursive) calls made to $\GMM(G, \pi, e)$ as a result of calling $\GMM(G, \pi, e')$ once for every edge $e'$. Then, for \textbf{any} edge $e$, it holds that
    $$
    \Exp_{\pi}[Q(G, \pi, e)] \leq 1 + 2\Delta,
    $$
    where $\Delta$ is the maximum degree of the graph, and the expectation is over the ordering $\pi$.
\end{lemma}

\begin{algorithm}
    \caption{$\textsc{GMM}(G, \pi, e)$ \cite{YoshidaYI09}}  
    \label{alg:low-deg-gmm}
    \KwIn{Graph $G = (V, E)$, edge ranks $\pi: E \to [0, 1]$, and an edge $e$}
    \KwOut{A boolean denoting whether $e$ appears in the greedy maximal matching corresponding to $\pi$ }
    \BlankLine
    \For{$e'$ adjacent to $e$, with $\pi_{e'} < \pi_e$, in increasing order of $\pi_{e'}$}{
        \If{$\textsc{GMM}(G, \pi, e')$}{
            \Return{\False}
        }
    }
    \Return{\True}
\end{algorithm}

With the \RGMM{} subroutine at hand, we move on to the implementation of \cref{alg:low-degree-matching}, outlined as \cref{alg:low-deg-vertex-value,alg:low-deg-edge-value}.
For the maximal matching, in each of the $1/\eps$ iterations, we use an \RGMM{} on the active vertices.
We implement this by adding another layer of recursion to \cref{alg:low-deg-gmm}, on the iteration number.
More accurately, the algorithm consists of two recursive functions $\EdgeValue(e, t)$ and $\VertexValue(v, t)$ which compute the updated saturation values after $t$ iterations, for an edge $e$ or a vertex $v$, respectively.
To compute the value of a vertex $v$ after $t$ iterations, we simply compute the value of the edges adjacent to it and add them to the initial saturation value $x_v$.
To compute the value of an edge after $t$ iterations, we need to determine two things: (1) its value after $t - 1$ iterations, and (2) whether it appears in the $t$-th maximal matching (if it does, the value shall increase by $\epsilon$).
To assert whether an edge appears in the $t$-th maximal matching, we use the recursive algorithm for $\RGMM$ on the active vertices.
Here, the active neighbors of a vertex are determined on the fly by recursively computing the vertex values after $t-1$ iterations, and checking if they are at most $1 - 2\epsilon$.

\begin{algorithm}
    \DontPrintSemicolon
    \caption{\VertexValue$(G, U, \vec{x}, v, t)$}
    \label{alg:low-deg-vertex-value}
    \KwIn{Graph $G = V(V, E)$, active set $U$, saturation vector $\vec{x}$, vertex $v$, and the number of iterations $t$}
    \KwOut{The updated saturation value of $v$, after $t$ iterations of $\textsc{MatchLowDegree}(G, U, \vec{x})$ (i.e., \cref{alg:low-degree-matching})}
    \BlankLine
    $x'_v \gets x_v$
    
    \For{$u \in N(v) \cap U$}{
        $x'_v \gets x'_v + \EdgeValue(G, U, \vec{x}, (u, v), t)$
    }

    \Return{$x'_v$}
\end{algorithm}

\begin{algorithm}
    \DontPrintSemicolon
    \caption{$\EdgeValue(G, U, \vec{x}, e, t)$}
    \label{alg:low-deg-edge-value}
    \KwIn{Graph $G = V(V, E)$, active set $U$, saturation vector $\vec{x}$, edge $e \in G[U]$, and the number of iterations $t$}
    \KwOut{The fractional value of $e$, after $t$ iterations of $\textsc{MatchLowDegree}(G, U, \vec{x})$ (i.e., \cref{alg:low-degree-matching}), and a boolean denoting whether $e$ appears in the $t$-th maximal matching}
    \BlankLine
    \uIf(\tcp*[f]{base case}){$t = 0$}{
        \Return{$0$, \False}
    }
    \Else{
        \tcc{Recursively recover the state of $e$ after $t - 1$ iterations, and implement $\RGMM$ if the endpoints of $e$ have not been saturated}
        
        $u, v \gets $ endpoints of $e$

        $x^{(t-1)}_e \gets \EdgeValue(G, U, \vec{x}, e, t-1)$

        $x^{(t-1)}_u \gets \VertexValue(G, U, \vec{x}, u, t - 1)$
        
        $x^{(t-1)}_v \gets \VertexValue(G, U, \vec{x}, v, t - 1)$

        \If(\tcp*[f]{$e$ was deactivated in the previous iterations}){$x^{(t-1)}_u \geq 1 - 2\epsilon$ and $x^{(t-1)}_v \geq 1 - 2\epsilon$}{
            \Return{$x^{(t-1)}_e$, \False}
        }

        \BlankLine
        \BlankLine
    
        $N(e) \gets \{e' \in G[U] \mid e' \text{ shares an endpoint with } e\}$ \tcp*{query adjacency lists}

        \For(\tcp*[f]{assign ranks to the neighboring edges}){$e' \in N(e)$}{
            \If{$e'$ has not already been assigned a rank}{
                Draw $\pi^{(t)}_{e'}$ uniformly from $[0, 1]$
            }
        }

        \BlankLine
        \BlankLine
        
        \For{$e' \in N(e)$, such that $\pi^{(t)}_{e'} < \pi^{(t)}_e$, in increasing order of $\pi^{(t)}_{e'}$}{
            $x^{(t)}_{e'}, b \gets \EdgeValue(G, U, \vec{x}, e', t)$

            \If(\tcp*[f]{$e'$ appears in the $\GMM$, hence $e$ does not}){$b = \True$}{
                \Return{$x^{(t-1)}_{e}$, \False}
            }
        }
        
        \Return{$x^{(t-1)}_{e} + \epsilon$, \True}\tcp*{no lower-rank neighboring edges appear in the \GMM, hence $e$ does}
    }
\end{algorithm}

The proof of correctness is straightforward.
For the query complexity, note that the guarantee of \cref{lem:yyi} bounds the number of in-queries for any edge $e$ by $O(\Delta)$, in expectation, when we call $\GMM(e')$ for a uniformly random edge $e'$.
While this bounds the number of (out-)queries required for computing $\GMM(e')$ for an average edge $e'$, certain edges might have a much larger expected number of out-queries.
Here, the number of out-queries of $e'$ refers to the total number of recursive calls made as a result of invoking $\GMM(e')$.
The main challenge is that with the added recursion on the iteration $t$,
a uniform query to the $t$-th level might result in non-uniform queries to the $(t-1)$-th level.
That is, as a result of querying $\EdgeValue(e, t)$ for a uniformly random $e$, the queries to $\EdgeValue(\cdot, t-1)$ may be concentrated on an edge $e'$ for which the number of expected number of out-queries is significantly larger than $\Delta$.
As a result, \cref{lem:yyi} does not yield a trivial $\Delta^{O(1/\epsilon)}$ bound for the query complexity of \EdgeValue.
The key intuition here is that while the queries made to the lower levels may be non-uniform, they are not far from uniform.
Loosely speaking, the distribution is distorted by a factor of $\poly(\Delta)$ with each level of recursion, which leads to a $\Delta^{O(1/\epsilon)}$ query complexity.
The proofs are formalized below.

\begin{proof}[Proof of \cref{lem:low-deg-implementation}]
    For the sake of brevity, we elide the inputs $G$, $U$, and $\vec{x}$ of $\VertexValue$ and $\EdgeValue$. We also note that $\EdgeValue$ is always called on edges that are initially active, i.e., they are in $G[U]$.
    As such, we only consider those edges in the analysis below.
    
    \paragraph{Correctness.} First, we prove correctness of these routines, i.e., that $\VertexValue(v, t)$ always correctly computes the updated saturation value, and $\EdgeValue(e, t)$ always correctly computes the updated value and whether $e$ appears in the $t$-th selected maximal matching (here, correctness is defined with respect to sequentially selecting $t$ \RGMM s in the active graph, as in \cref{alg:low-degree-matching}).
    This is established by induction on $t$.
    For the base case $t = 0$, the claim holds trivially, as $\EdgeValue(e, 0)$ returns $(0, \False)$, and $\VertexValue(v, 0)$ returns the initial value $x_v$.

    For $t > 0$, note that the correctness of $\VertexValue(v, t)$ is implied by the correctness of $\EdgeValue(e, t)$, since $\VertexValue$ simply computes the updated values of the neighboring edges and adds them to the initial saturation $x_v$.
    Thus, it suffices to prove the correctness of $\EdgeValue(e, t)$.
    By the induction hypothesis, the recursive calls of $\EdgeValue(e, t)$ to $\EdgeValue(\cdot, t - 1)$ and $\VertexValue(\cdot, t - 1)$ yield the correct output.
    As a result, $\EdgeValue(v, t)$ can correctly determine the active subgraph after $t - 1$ iterations, and compute the $\RGMM$ corresponding $\pi^{(t)}$.
    This concludes the proof of correctness.

    \paragraph{Query Complexity.} Moving on to the query complexity, we introduce some notation.
    Let $\pi^{(<t)}$ denote ranks used for computing the first $t - 1$ \RGMM s. For the sake of analysis, we can assume that all the ranks are drawn for the entire graph at the start of the algorithm.
    Consider the process of drawing all the ranks and invoking $\VertexValue(v, 1/\epsilon - 1)$ once for each vertex $v$.
    Let $Q(e', t)$, \emph{the in-degree of an edge $e'$ at level $t$}, be a random variable denoting the number of recursive calls made to $\EdgeValue(e', t)$ as the result of the aforementioned process.

    We claim that it suffices to prove $\Exp[Q(e, 0)] = \Delta^{O(1/\epsilon)}$ for all $e$.
    To see this, first note that the adjacency lists are accessed at most twice in each call to $\EdgeValue$ (once for each endpoint, to obtain all the neighboring edges), and once in each call to $\VertexValue$ (to obtain the adjacent edges).
    Therefore, it suffices to bound the total number of (recursive) calls to $\VertexValue$ and $\EdgeValue$ by $n \Delta^{O(1/\epsilon)}$.
    Furthermore, except the initial calls to $\VertexValue(v, 1/\epsilon - 1)$, every call to $\VertexValue(\cdot, t)$ is triggered by a call to $\EdgeValue(\cdot, t + 1)$, and each call to $\EdgeValue(\cdot, t + 1)$ is responsible for at most two of the calls to $\VertexValue(\cdot, t)$.
    Therefore, it suffices to bound the total number of calls just to $\EdgeValue$ by $n \Delta^{O(1/\epsilon)}$.
    Additionally, for $t > 0$, each call to $\EdgeValue(e, t)$ makes a call to $\EdgeValue(e, t-1)$.
    Hence, since there are $1/\epsilon$ levels, it suffices to bound the total number of calls made to $\EdgeValue(\cdot, 0)$ by $n \Delta^{O(1/\epsilon)}$.
    Finally, note that there are at most $m \leq n\Delta$ edges.
    As a result, it suffices to bound the number of calls to $\EdgeValue(e, 0)$ by $\Delta^{O(1/\epsilon)}$ for each edge $e$.

    To prove $\Exp[Q(e, 0)] = \Delta^{O(1/\epsilon)}$, we use a downward induction on $t = 1/\eps - i$, and show $$\Exp[Q(e, 1/\epsilon - i)] \leq (4\Delta)^{2i}.$$
    In fact, we prove a stronger claim that this holds even when we condition on an arbitrary set of ranks $\pi^{(<t)}$ (i.e., no matter what the graph is after the first $t - 1$ iterations).
    The base case, $t = 1/\epsilon - 1$, follows directly from \cref{lem:yyi}.
    Conditioning on an arbitrary set of ranks $\pi^{(<t)}$, yields a set of active vertices $U_t$ at the beginning of the $t$-th iteration.
    We are drawing a random rank $\pi^{(t)}$ for the $t$-th level,
    and invoking $\VertexValue(v, t)$ once for each $v$.
    These result in direct calls to $\EdgeValue(e, t)$ twice for each edge $e$.
    Then, $\EdgeValue(e, t)$ essentially simulate \RGMM{} (\cref{alg:low-deg-gmm}) on $G[U_t]$.
    Therefore, by \cref{lem:yyi}, we can conclude:
    $$
    \Exp[Q(e, t) \mid \pi^{(<t)}] \leq 2(1 + 2\Delta) \leq (4\Delta)^2,
    $$
    for all $e$.

    For the induction step, take $t = 1/\eps - i$ for $i > 1$, fix a set of ranks $\pi^{(<t)}$, and assume that the hypothesis holds for $t + 1$,
    i.e.\ for any $\pi^{(<t+1)}$, we have
    $$
    \Exp[Q(e, t + 1) \mid \pi^{(<t + 1)}] \leq (4\Delta)^{2(i - 1)},
    $$
    for all $e$, where we are specifically interested in $\pi^{(<t + 1)}$ that extend $\pi^{(<t)}$ by adding one level $\pi^{(t)}$.
    Lifting the condition on $\pi^{(t)}$, the above implies:
    $$
    \Exp[Q(e, t + 1) \mid \pi^{(<t)}] \leq (4\Delta)^{2(i - 1)}.
    $$
    We consider the calls initially made to level $t$ as a result of the calls to level $t + 1$.
    Each call to $\EdgeValue(e, t + 1)$ makes a call to $\EdgeValue(e, t)$ and two calls to $\VertexValue(\cdot, t)$, one for each endpoint.
    The calls to $\VertexValue(\cdot, t)$, in turn, make a call to $\EdgeValue(e', t)$ for each adjacent edge $e'$. We refer to these as the initial calls to level $\EdgeValue(\cdot, t)$.
    Combining with the induction hypothesis, the number of initial calls made to $\EdgeValue(e, t)$ can be bounded by $(1 + 2\Delta)(4\Delta)^{2(i - 1)} \leq (4\Delta)^{2i -1}$ in expectation for any $e$.
    
    The key step here is to artificially increase the expected number of calls to level $t$ to reach $(4\Delta)^{2i -1}$.\footnote{This is parallel to Lemmas 3.4 - 3.6 of \cite{YoshidaYI09} for controlling the distortion in the uniform distribution.}
    That is, if the expectation is lower for any edge $e$, we add extra queries to $\EdgeValue(e, t)$.
    With each $\EdgeValue(e, t)$ called exactly $(4\Delta)^{2i - 1}$ times in expectation, we can invoke \cref{lem:yyi} again to obtain:
    $$
    \Exp[Q(e, t) \mid \pi^{(<t)}] \leq (4\Delta)^{2i - 1} \cdot (1 + 2\Delta) \leq (4\Delta)^{2i}.
    $$
    This concludes the proof of the induction and the claim.
\end{proof}

\subsection{Proof of \Cref{thm:subset-query}}

We will now combine the bounds of \Cref{lem:degree-estimation-subset-queries}
and \Cref{lem:low-deg-implementation} to give a bound on the number of subset queries required to implement \textsc{EstimateMatchingSize} in its entirety.
From this, we will conclude \Cref{thm:subset-query}. 

First---primarily because the streaming algorithms will use a different approach for the low-degree phase---we 
will just bound the number of subset queries made during the high-degree phase of the algorithm, i.e. before \textsc{MatchLowDegree} is called.
\begin{claim}
\label{clm:subset-query-bound}
\textsc{EstimateMatchingSize}, outside of \textsc{MatchLowDegree},
makes $\poly(\log n, 1/\eps)$ subset queries.
\end{claim}
\begin{proof}
Subset queries outside of \textsc{MatchLowDegree} are only made within \textsc{ComputeDegrees}, so we will bound the number of calls to \textsc{ComputeDegrees}.
Each call to \textsc{PeelHighDegree} makes $O(1/\eps)$ calls to \textsc{ComputeDegrees} (see proof of \Cref{lem:fractional-validity} for this argument), and there are 
$\le \frac{\log n}{\log((1-\eps)^{-1})}$ calls to \textsc{PeelHighDegree}.
Thus we can bound the total number of calls to \textsc{ComputeDegrees} by $O((\log n)/\eps)$. 
By \Cref{lem:degree-estimation-subset-queries}, each call to \textsc{ComputeDegrees} makes $O((\log (n/\delta) \log n)/\eps^3)$ subset queries,
where $\delta$ is the failure probability of \textsc{ComputeDegrees}.
For our algorithm to have a global failure probability of $1/n$, 
we will let $\delta := c \cdot \eps/n \log n$, where $c$ is determined by the hidden constant in the $O((\log n) /\eps)$ bound on the number of calls
to \textsc{ComputeDegrees}.

This gives a final subset-query complexity of 
\[O\left(\frac{\log n}{\eps} \cdot \frac{\log (n/\eps) \log n}{\eps^3} \right) = \poly(\log n, 1/\eps). \qedhere\]
\end{proof}

Now we will combine this with the bounds on the complexity of \textsc{MatchLowDegree} to conclude \Cref{thm:subset-query}.
In fact we will show a slight relaxation of the theorem with additive-multiplicative approximation and expected query complexity; standard arguments show that the stronger version follows from that.
\begin{theorem}[Relaxation of \Cref{thm:subset-query}]
Given an $n$-vertex graph $G$, for any $\epsilon > 0$, there is a randomized algorithm that $(2, \epsilon n)$-approximates the size of maximum matching in $G$ w.h.p. The expected subset-query complexity is $\poly(\log n) \cdot (1/\epsilon)^{O(1/\eps)}$.
\end{theorem}
\begin{proof}
We will call \textsc{EstimateMatchingSize}$(G,\eps, \eta=2/\eps^2)$ using the subset-query implementations of \textsc{ComputeDegrees} and \textsc{MatchLowDegree},
where instead of returning the exact value 
\[\hat{\mu} := \frac{1-\eps}{1+\eps} \cdot \frac 12 \sum_{v \in V} x_v,\]
we will sample a random set $S$ of $O(\tfrac{\log n}{\eps^2})$ indices and compute 
\[\hat{\mu}_S :=  \frac{1-\eps}{1+\eps} \cdot \frac {n}{2|S|} \sum_{v \in S} x_v, \]
then return $\hat{\mu}_S' := \hat{\mu}_S - \eps n$.
First we observe that we can choose the constant such that 
\[\Pr[|\hat{\mu} - \hat{\mu}_S| > \eps n] \le n^{-c}.\]
This follows from the fact that $x_v \in [0,1]$ for all $v$ and a Hoeffding bound.
Combined with the validity and approximation guarantees of \Cref{cor:validity,lem:apx-factor},
this yields an estimate $\hat{\mu}_S'$ such that
\[\frac{\mu(G)}{2+O(\eps n)} - 2\eps n \le \hat{\mu}_S' \le \mu(G).\]
Since $\mu(G) < n$, we conclude the multiplicative-additive approximation guarantee:
\[\mu(G)/2 - O(\eps n) \le \hat{\mu}_S' \le \mu(G),\]
where the $O(\eps n)$ can be reduced to $\eps n$ by scaling $\eps$ by the appropriate constant.

Now we bound the query complexity. 
For $v \in S$ such that $v$ is removed from $U$ during a call to \textsc{PeelHighDegree}, $x_v$ can just be retrieved from the vector $\vec{x}$ 
without making a call to the local implementation of \textsc{MatchLowDegree}.
For $v$ that remains in $U$ after all calls to \textsc{PeelHighDegree}, \textsc{MatchLowDegree} is called on a graph of degree at most $\eta =O(1/\eps^2)$, 
so by \Cref{lem:low-deg-implementation}, $x_v$ takes $(1/\eps)^{O(1/\eps)}$ subset queries to compute, in expectation over a random vertex.

Thus, combined with \Cref{clm:subset-query-bound}, the expected total number of subset queries is bounded by
\[\poly(\log n) \cdot (1/\epsilon)^{O(1/\eps)},\]
which concludes the proof.
\end{proof}

\begin{remark}
\label{remark:strengthening}
To conclude the stronger version of the theorem, we restart the algorithm if it has not returned within twice its expected running time, and make at most $O(\log n)$ attempts.
This gives an algorithm that still succeeds w.h.p.\ and always has $\poly(\log n) \cdot (1/\eps)^{O(1/\eps)}$ query complexity.

To go from a multiplicative-additive to a multiplicative approximation, we first get rid of singleton vertices in the low-degree case. After that, the maximum matching of the remaining vertices will be of size at least $\Omega(k/\Delta)$ where $\Delta=O(1/\epsilon^2)$ is the maximum degree in the low-degree case and $k$ is the number of non-singletons. 
Sampling $\log n \cdot (\eps/\Delta)^{-2}$ vertices instead of $
\log n/ \eps^2$ gives a multiplicative approximation. 
This technique is standard (see \cite{Behnezhad21}).



\end{remark}

\section{Streaming Distance-to-Monotonicity Estimation}\label{sec:streaming-monotonicity}

In this section we prove \Cref{thm:main}, giving a $\tilde{O}(n)$-space, $\sqrt{n}^{1 + o(1)}$-pass streaming algorithm for distance-to-monotonicity estimation.
We do this by giving a streaming implementation of a subset query in the violation
graph $G_f$,
which will allow us to use the subset-query algorithm to estimate the size of the 
maximum matching in the violation graph.
From \Cref{fact:matching-dmon}, it will follow that this gives an estimator for 
distance to monotonicity.

\begin{lemma}[Streaming implementation of subset query]
\label{lem:subset-query-streaming}
A subset query in the conflict graph $G_f$ can be implemented in $\sqrt{n}^{1 + o(1)}$ passes over the edge and label stream of $G, f$ and $\tilde{O}(n)$ space.
The query succeeds with high probability.
\end{lemma}
We cite the following result, originally presented by \cite{JLS19} as a PRAM algorithm.
See Proposition 4 of \cite{AJJST22} and the discussion beneath it for an explanation of how it is implemented in the streaming model.
\begin{proposition}[Streaming shortcut set computation \cite{JLS19, AJJST22}]
\label{prop:hopset}
There is a $\sqrt{n}^{1 + o(1)}$-pass, $\tilde{O}(n)$-space algorithm that given directed $G=(V,E)$, with high probability computes a \emph{shortcut set}\footnote{Some authors make a distinction between a shortcut set and a hopset, where the latter consists of additional \emph{weighted} edges which approximately preserve shortest paths.  For clarity, we use the term ``shortcut set'' for the unweighted concept instead of ``hopset'' as in \cite{JLS19, AJJST22}.} $H \subset V \times V$ of $\tilde{O}(n)$ additional edges such that 
\begin{enumerate}
    \item $\mathsf{diam}(G \cup H) \le \sqrt{n}^{1 + o(1)}$, and
    \item For all $(u,v) \in V \times V$, $v$ is reachable from $u$ in $G \cup H$ iff $v$ is reachable from $u$ in $G$.
\end{enumerate}
\end{proposition}

\begin{proof}[Proof of \Cref{lem:subset-query-streaming}]
Let $S$ be the queried subset. 
We will first use one pass and $O(n)$ space to store all the labels of $f$. 
Partition $S$ into $S_0 := S \cap f^{-1}(0)$ and $S_1 := S \cap f^{-1}(1)$.
We will then use $\sqrt{n}^{1 + o(1)}$ passes and $\tilde{O}(n)$ space to build a $\sqrt{n}^{1 + o(1)}$-diameter shortcut set as specified in \Cref{prop:hopset}. 
We will run a (directed) parallel BFS in $G \cup H$ from $S_1$ as follows: 
\begin{enumerate}
    \item For one pass, whenever a (directed) edge $(u,v)$ appears in the stream where $u \in S_1$, place a mark on $v$.
    \item For $\mathsf{diam}(G \cup H)$ passes, whenever a (directed) edge $(u,v)$ appears in the stream where $u$ is marked, mark $v$.
    \item Let $T_1$ be the set of marked vertices labeled 0.
\end{enumerate}
We will then run the analogous BFS from $S_0$:
\begin{enumerate}
    \item For one pass, whenever a (directed) edge $(u,v)$ appears in the stream where $v \in S_0$, place a mark on $u$.
    \item For $\mathsf{diam}(G \cup H)$ passes, whenever an edge $(u,v)$ appears in the stream where $v$ is marked, mark $u$.
    \item Let $T_0$ be the set of marked vertices labeled 1.
\end{enumerate}
First, observe that the BFS marks exactly the set of descendants of $S_1$ and ancestors of $S_0$ in $G \cup H$.
By the fact that $H$ preserves the reachability relation in $G$, these are also exactly the descendants and ancestors in $G$.
A vertex $u$ violates monotonicity with $v \in S_0$ iff $f(u)=1$ and $u$ is an ancestor of $v$, so $u$ violates monotonicity with some $v \in S_0$ iff $u \in T_0$.
Similarly, $u$ violates monotonicity with $v \in S_1$ iff $f(u)=0$ and $u$ is a descendant of $v$,
so $u$ violates monotonicity with some $v \in S_0$ iff $u \in T_0$.

Thus, $T_0 \cup T_1$ is the correct answer to the subset query $S=S_0 \cup S_1$; i.e. it contains exactly the set of vertices adjacent to $S$ in the conflict graph.

Since $\mathsf{diam}(G \cup H) \le \sqrt{n}^{1 + o(1)}$ and we only need to store the marks, the number of passes required for the BFS is $\le \sqrt{n}^{1 + o(1)}$ and the space is $O(n)$.
This gives a total of $\sqrt{n}^{1 + o(1)}$ passes and $\tilde{O}(n)$ space for the entire subset query.
\end{proof}

We will now complete the proof of \Cref{thm:main} by combining the pass complexity
of this subset query implementation with the bound on the subset-query complexity.
Rather than using the subset-query implementation exactly as presented in \Cref{sec:subset-query-implementation}, we simplify the low-degree stage and 
improve the $\eps$ dependence by taking advantage of the fact that once the 
degree falls low enough, we can build an explicit representation of $G_f$ in one pass.
So instead of using the local algorithm which makes $(1/\eps)^{O(1/\eps)}$
subset queries, we will build $G_f$ and use the global algorithm \Cref{alg:low-degree-matching} as written.


\begin{proof}[Proof of \Cref{thm:main}]
We will run \textsc{EstimateMatchingSize} in the violation graph $G_f$ using the subset-query implementation of \textsc{ComputeDegrees}, 
and a single-pass global implementation of \textsc{MatchLowDegree}. Since only these functions require access to the edges of $G_f$,
no passes are made outside of them. 

We will first bound the space requirements of \textsc{EstimateMatchingSize} and \textsc{PeelHighDegree}
outside of the calls to \textsc{ComputeDegrees} and \textsc{MatchLowDegree}.
We maintain $\dmax$, which requires $O(\log n)$ space, the set of active vertices $U$, which requires $O(n)$ space,
and the saturation vector $\vec{x}$, which requires $O(n \log(1/\eps))$ space to store weights of granularity $O(\eps)$.

We will now bound the space complexity of \textsc{MatchLowDegree}. We will build the entire conflict graph $G_f[U]$ in one pass by storing all the edges and adding an edge between every reachable pair that violates monotonicity.
Then we will run the rest of the subroutine on the stored graph, making no more passes.
The graph at this point has degree at most $O(1/\eps^2)$; 
thus the space complexity of the graph is $O(n \log n/\eps^2)$.
Each matching takes $O(n \log n)$ space, but the matchings are not stored between iterations --- only $U$ and $\vec{x}$ persist, which take $O(n \log (1/\eps))$ space.
This subroutine then takes $O(n \log n/\eps^2)$ space in total.

Now we will analyze the space and pass complexity of \textsc{ComputeDegrees}.
By \Cref{lem:subset-query-streaming}, \textsc{ComputeDegrees} takes $O(n \log(n/\eps))$ space, so the total space complexity is still $\tilde{O}(n/\eps^2))$.

By \Cref{clm:subset-query-bound}, the number of subset queries is bounded by $\poly(\log n, 1/\eps)$.
Combined with \Cref{lem:subset-query-streaming}, this gives a total pass complexity of $\sqrt{n}^{1 + o(1)} \cdot \poly(1/\eps)$.
By \Cref{lem:fractional-validity,lem:apx-factor}, the output is a $(2+O(\eps))$-approximation to the maximum matching size in $G_f$;
the standard technique of rescaling $\eps$ by the appropriate constant reduces this to a $(2+\eps)$-approximation.
By \Cref{fact:matching-dmon}, it is thus a $(2+\eps)$-approximation to the distance to monotonicity,
from which the theorem follows.
\end{proof}

\section{Implementation With Vertex Queries}\label{sec:vertex-query-implementation}

In this section we show how to adapt the ideas of \Cref{sec:main-algorithm,sec:subset-query-implementation} to the vertex-query model.  The main difficulty in this model, compared to the subset-query model, is that accurate multiplicative degree estimates are not possible for low-degree vertices:  for instance, identifying the set of degree-zero vertices in a graph would require $\Theta(n)$ vertex queries.  This rules out directly implementing \Cref{assume:compute-degrees}.  Instead, we use the following alternate assumption:

\begin{assumption}[\textsc{ComputeDegrees}$_\theta$]
\label{assume:compute-degrees-theta}
There is a function $\textsc{ComputeDegrees}_{\theta}(G,U, \eps)$ that takes an arbitrary graph $G=(V,E)$, 
subset $U \subseteq V$, and approximation parameter $\eps > 0$, and returns a vector $\vec{d}$ such that:
\begin{itemize}
\item $d_v$ is nonzero for all $v \in V$ with $|N(v) \cap U| \ge \theta$, and zero for all $v \in V$ with $|N(v) \cap U| \le (1 - \eps)\theta$.
\item If $d_v$ is nonzero, then $d_v = (1 \pm \eps)|N(v) \cap U|$.
\end{itemize}
\end{assumption}

We will show in \Cref{lem:degree-estimation-vertex-queries} that it is possible
to implement $\textsc{ComputeDegrees}_\theta$ for $\theta \approx \sqrt{n}$ 
using about $\sqrt{n}$ vertex queries.
This will allow us to run \textsc{PeelHighDegree} with degree thresholds above $\sqrt{n}$; however it will not allow us to drop the degree all the way to $2/\eps^2$ as we do in the subset query implementation.
Since we must transition to the low-degree phase at a threshold $\eta \approx \sqrt{n}$, 
it will also no longer be feasible to use the result of \Cref{lem:low-deg-implementation},
as each $x_v$ would take $\eta^{O(1/\eps)} = n^{O(1/\eps)}$ queries to compute.

To eliminate the $O(1/\eps)$ in the exponent, we switch from computing an adaptive sequence of $1/\eps$ maximal integral matchings to computing a single maximal fractional matching,
which increases the integrality gap from $1/(1-\eps)$ to $3/2$.
Thus, at the end of the algorithm, we will scale the estimate by $2/3$ instead of $1-\eps$, achieving an approximation ratio of 3. 

In summary, the modified global algorithm \textsc{EstimateMatchingSizeVQ} differs from \Cref{alg:matching} in four ways:
\begin{enumerate}
\item The transition from high to low degrees is defined as $\eta = \sqrt{n}$ instead of $2/\eps^2$.
\item The \textsc{PeelHighDegree} subroutine uses \textsc{ComputeDegree}$_\theta$ with parameter $\theta = \frac{\eps^2 \eta}{2\ln n}$ in place of \textsc{ComputeDegree}.
\item The \textsc{MatchLowDegree} subroutine is replaced with \textsc{MatchLowDegreeVQ} defined in \Cref{alg:low-degree-matching-vq}
\item The final estimate returned by the algorithm is $\hat\mu = \tfrac{1}{1 + 3\eps} \cdot \tfrac13 \sum_v x_v$ instead of $\tfrac{1 - \eps}{1 + \eps} \cdot \tfrac12 \sum_v x_v$. 
\end{enumerate}

\begin{algorithm}[H]
\DontPrintSemicolon
\caption{$\textsc{MatchLowDegreeVQ}(G,U,\vec{x}, \eps)$}
\label{alg:low-degree-matching-vq}
\KwIn{Graph $G = (V,E)$, active set $U$, saturation vector $\vec{x}$, approximation parameter $\eps > 0$}
\KwOut{Updated saturation vector}
\BlankLine
$V' \gets V \times [1/\eps]$\;
$E' \gets \{\{(v_1, i_1), (v_2, i_2)\} : \{v_1, v_2\} \in G\}$\;
$G' \gets (V', E')$\;
$U' \gets \{(v, i) \in V' \times \{0, \ldots, 1/\eps\} : v \in U \land \eps i \le x_v\}$\;
$M \gets \textsc{MaximalMatching}(G', U')$\;
\For{$\{(u, i_1), (v, i_2)\} \in M$}{
    $x_u \gets x_u + \eps$\;
    $x_v \gets x_v + \eps$\;
    \tcc{$y_{uv} \gets y_{uv} +\eps$}
}
\Return $\vec{x}$
\end{algorithm}

\subsection{Correctness and Approximation Ratio of the Modified Global Algorithm}

Now we argue that the correctness conditions of \Cref{sec:main-algorithm} apply to the modified algorithm, with slight changes.

\begin{lemma}\label{lem:high-degree-validity-vq}
    After the last call to \textsc{PeelHighDegree} in \textsc{EstimateMatchingSizeVQ},
    the vectors $\vec{x}, \vec{y}$ satisfy the three conditions of \Cref{lem:high-degree-validity},
    except that the first condition is replaced by $(1-\eps) \left[ \sum_{u \in N(v)} y_{uv}\right] - \eps \le x_v \le (1+\eps) \sum_{u \in N(v)} y_{uv}$.
\end{lemma}
\begin{proof}
    The behavior of \textsc{PeelHighDegree} is almost the same under \Cref{assume:compute-degrees-theta} instead of \Cref{assume:compute-degrees}.
    The identity of vertices in $H$ is unchanged, as is the approximation guarantee on $(d_U)_v$ for $v \in H$.  The only difference is that $(d_H)_v$ may be 0 for vertices $v \in U \setminus H$ with fewer than $\theta$ neighbors in $U$.
    For those vertices, the value of $x_v$ may be too small, as compared to the original \textsc{PeelHighDegree}.
    It immediately follows that conditions 2 and 3 as well as the upper bound of condition 1 are satisfied.

    Now we bound the amount of error in $x_v$.
    The size of the error accumulates by at most $\frac{\eps}{2} \cdot \frac{\theta}{\dmax} \le \frac{\eps\theta}{2\eta}$ in each of at most $4/\eps$ iterations of \textsc{PeelHighDegree},
    and thus the cumulative error across at most $\frac{\ln n}{\eps}$ calls to \textsc{PeelHighDegree} is bounded by $\frac{2\theta \ln n}{\eps \eta}$, which is at most $\eps$ by choice of $\theta$.
\end{proof}

\begin{lemma}[VQ analogue of \Cref{cor:validity} (validity)]
\label{lem:fractional-validity-vq}
Assume $\eps \le \tfrac13$, and  let $\vec{x}$ be the vector computed by \textsc{EstimateMatchingSizeVQ} (\Cref{alg:matching}). 
Then there exists a matching in $G$ of value 
\[\hat\mu = \frac{1}{1 + 3\eps} \cdot \frac13 \sum_v x_v.\]
\end{lemma}
\begin{proof}
First we will show that there exists a fractional matching in $G$ with vertex saturation vector $\frac{1}{1+3\eps} \vec{x}$.
It follows from \Cref{lem:high-degree-validity-vq} and the behavior of \Cref{alg:low-degree-matching-vq} that the final values $\vec{x}$ and $\vec{y}$ satisfy $\frac{1}{1+\eps} x_v \le \sum_{u \in N(v)} y_{uv} \le 1 + \eps$ for every $v$;
thus $\frac{1}{(1+\eps)} \vec{y}$ is a fractional matching with vertex saturations at least $\frac{1}{(1+\eps)^2} \vec{x} \ge \frac{1}{1 + 3\eps} \vec{x}$.

This fractional matching has value at least $\tfrac{1}{1 + 3\eps} \cdot \tfrac 12 \sum_v x_v$. 
By \Cref{fact:frac-gap-three-halves}, there is then an integral matching of value at least
$$\frac 23 \cdot \frac{1}{1 + 3\eps} \cdot \frac 12 \sum_v x_v = \frac{1}{1 + 3\eps} \cdot \frac13 \sum_v x_v.$$
\end{proof}

\begin{lemma}[VQ analogue of \Cref{lem:apx-factor} (approximation ratio)]
\label{lem:apx-factor-vq}
The fractional matching value $\hat{\mu}$ returned by \textsc{EstimateMatchingSizeVQ} satisfies 
\[\hat{\mu} \ge \frac{\mu(G)}{3 + O(\eps)}.\]
\end{lemma}

\begin{proof}
We observe that every edge $(u,v) \in G$ satisfies $x_u + x_v \ge 1 - O(\eps)$.
Let $L$ be the set of vertices in $U$ when \textsc{MatchLowDegreeVQ} is called; the case where $u$ or $v$ is in $V \setminus L$ is identical to the analogous case in the proof of \Cref{lem:apx-factor}.
When $u$ and $v$ are both in $L$, this condition follows from the maximality of $M$ in \textsc{MatchLowDegreeVQ}.  The lemma statement then follows from \Cref{clm:apx-factor}.
\end{proof}

\subsection{Implementing the Modified Algorithm Using Vertex Queries}

We start by showing that \Cref{assume:compute-degrees-theta} can be implemented in the vertex query model.

\SetKwFor{RepTimes}{repeat}{times}{end}
\begin{algorithm}[H]
    \DontPrintSemicolon
    \caption{Implementing \textsc{ComputeDegrees}$_\theta$ via vertex queries}
    \label{alg:deg-vertex}
    \KwIn{
        Graph $G = (V, E)$ with $|V| = n$, subset $U \subseteq V$, threshold $\theta > 0$, approximation parameter $\eps > 0$, failure probability $\delta > 0$
    }
    \KwOut{
        Degree estimates $\vec{d}$
    }
    \BlankLine
    $r \gets C\eps^{-2}\frac{|U|}{\theta}\log(n \delta^{-1})$ for a sufficiently large constant $C$\;
    $\vec{X} \gets 0$\;
    \RepTimes{$r$}{
      Choose $u \in U$ uniformly at random\;
     \For{$v \in N(u)$}{
        $X_v \gets X_v + 1$
      }
    }
    \For{$v \in V$}{
    $d_v \gets \frac{|U|}{r} X_v$\;
    \If{$d_v \le (1 - \eps/2)\theta$}{
    $d_v \gets 0$\;
    }
    }
    \Return{$\vec{d}$}
\end{algorithm}

\begin{lemma}[Degree estimation via vertex queries] \label{lem:degree-estimation-vertex-queries}
The functionality of \Cref{assume:compute-degrees-theta} can be implemented using $O\left(\eps^{-2} \frac{|U|}{\theta} \log(n \delta^{-1}) \right)$ vertex queries, with a failure probability of $\delta$.
\end{lemma}
\begin{proof}
  Refer to \Cref{alg:deg-vertex}.
  Fix a vertex $v \in V$ and consider $X_v$, the number of times $v$ appears as a neighbor of one of the queried vertices.  Then $\E X_v = \frac{r}{|U|} \cdot |N(v) \cap U|$.  
   If $|N(v) \cap U| \ge \theta/4$ then by a Chernoff bound,
  \begin{align*}
  \Pr[|X_v - \E X_v| \ge (\eps/3)|\E X_v|]
  &\le 2\exp\left(-\frac{(\eps/3)^2}{3} \E X_v\right) \\
  &\le n^{-1} \delta,
  \end{align*}
  provided $C$ is sufficiently large.
  On the other hand, if $|N(v) \cap U| \le \theta/4$ then
  \begin{align*}
  \Pr\left[X_v \ge \frac34\cdot\frac{r}{|U|}\cdot \theta\right]
  &\le \exp\left(-\frac{r}{|U|}\theta/4\right) \\
  &\le n^{-1} \delta,
  \end{align*}
  again provided $C$ is sufficiently large.  (This uses the fact that the one-sided Chernoff upper bound holds with $\E X_v$ replaced by the upper bound $\frac{r}{|U|}\theta/4$, which follows from a coupling argument.)
  Thus by taking a union bound over all $v \in V$, with probability at least $1 - \delta$
  the returned $\vec{d}$ correctly distinguishes vertices with at least $\theta$ neighbors in $U$ from those with at most $(1 - \eps)\theta$, and correctly estimates degrees whenever it returns a nonzero value.
\end{proof}


Now we use the following result of \cite{Behnezhad21} to implement query access to the output of \textsc{MaximalMatching}.

\begin{theorem}[{\cite[Theorem 3.5]{Behnezhad21}}]\label{lem:d-query-gmm}
Let $G=(V,E)$ be a graph with average degree $\bar{d}$. Let $\Pi$ be the set of all permutations over $E$. For a permutation $\pi \in \Pi$. For a vertex $v \in V$, there is a randomized oracle $\mathrm{VO}(v,\pi)$ that determines whether $v$  is matched in $\textsc{GMM}(G,\pi,e)$ (\cref{alg:low-deg-gmm}) and, if so, identifies the matching edge. Let $T(v,\pi)$ denote the total number of recursive calls to the edge oracle generated during the execution of  $\mathrm{VO}(v,\pi)$. Then for a vertex $v$ chosen uniformly at random from $V$ and a permutation $\pi$ chosen uniformly at random from $\Pi$, independently from $v$,
\[
    \mathbb{E}_{v \sim V,\, \pi \sim \Pi}[T(v,\pi)] = O(\bar{d} \cdot \log n).
\]
\end{theorem}

\begin{lemma}\label{lem:low-deg-matching-vertex}
There exists a randomized vertex-query algorithm that computes $\eps n \pm \sum_v x_v'$, where $x_v'$ is the output of \textsc{MatchLowDegreeVQ} (\Cref{alg:low-degree-matching-vq}) on inputs $G, U, \vec{x}, \eps$, using vertex queries to $G$.
The algorithm uses $O(\overline{d}_U \cdot \eps^{-3}(\log n)(\log \delta^{-1}))$ vertex queries in expectation and succeeds with probability $1 - \delta$,
where $\overline{d}_U$ is the average degree in $G[U]$.
\end{lemma}
\begin{proof}
    Let $G' = (V', E')$ and $U' \subset V'$ be as in \textsc{MatchLowDegreeVQ}.

    It suffices to estimate to within an additive error of $\eps$ the proportion $p$ of vertices in $V'$ which are matched by a maximal matching of $G'[U']$,
    since given such an estimate $\hat p$ we can output $\hat p n + \sum_{v \in V} x_v = \eps n \pm \sum_{v \in V} x_v'$.
    We can accomplish this by sampling $\Theta(\eps^{-2}\log \delta^{-1})$ random vertices and determining if they are matched.

    Draw a permutation $\pi$ uniformly at random over $E(G'[U'])$ and let $M$ be the the matching computed by the algorithm of \Cref{lem:d-query-gmm} on input $G'[U']$.
 
    Pick a vertex $v \in V'$ uniformly at random.  If $V' \notin U'$ then it is not matched, so assume it is in $U'$.  We run the oracle $\mathrm{VO}(v, \pi)$ on $G'[U']$, implementing the oracle's adjacency-list access to $G'[U']$ using vertex queries to $G$.
    The expected total number of vertex queries is then $O(\overline{d}_{G'[U']} \cdot \eps^{-2} (\log n)(\log \delta^{-1})) = O(\overline{d}_U \cdot \eps^{-3}(\log n)(\log \delta^{-1}))$.
 \end{proof}

\subsection{Proof of \Cref{thm:vertex-query}}
Finally, we will combine the results of the previous subsections and bound the total number of vertex queries, obtaining \Cref{thm:vertex-query}.
As before, we will bound the expected number of queries; the same repetition technique discussed in \Cref{remark:strengthening} applies.
(The multiplicative error technique does not apply due to the fact that we cannot efficiently detect singleton vertices using vertex queries.)

\begin{proof}[Proof of \Cref{thm:vertex-query}]
From \Cref{lem:low-deg-matching-vertex} we obtain an estimate $\hat{x} = \eps n \pm \sum_{v \in S} x_v$, from which we compute and return
\[\hat{\mu}_S :=  \frac{1}{1+3\eps} \cdot \frac13 \hat x - \eps n.\]
Then this is a $(3, \eps n)$-approximation according to
\Cref{lem:fractional-validity-vq,lem:apx-factor-vq}.

Now we bound the number of vertex queries in the high-degree phase. 
Vertex queries in this phase are only made during calls to \textsc{ComputeDegrees}.
By the argument in the proof of \Cref{clm:subset-query-bound},
\textsc{ComputeDegrees} is called $O((\log n)/\eps)$ times.
By \Cref{lem:degree-estimation-vertex-queries}, the query complexity of $\textsc{ComputeDegrees}_\theta$ with failure probability $n^{-c}$ is 
\[O\left( \frac{n \log n}{\eps^2 \theta} \right).\]
We are calling it with 
\[\theta = \frac{\eps^2 \sqrt{n}}{2 \ln n},\]
therefore the total query complexity of all calls to \textsc{ComputeDegrees} is
\[O\left( \frac{n \log^2 n}{\eps^4 \sqrt{n}} \cdot \frac{\log n}{\eps} \right) =O\left( \frac{\sqrt{n} \log^3 n}{\eps^5} \right) =  \sqrt{n} \cdot \poly(\log n, \eps^{-1}). \]

Finally, the number of vertex queries in the low-degree phase is bounded by $\overline{d}_L \cdot \poly(\log n, \eps^{-1}) \le \sqrt{n} \cdot \poly(\log n, \eps^{-1})$ according to \Cref{lem:low-deg-matching-vertex}, so this is the query complexity for the algorithm as a whole.
\end{proof}

\section{Better Approximation Ratios with Subset Queries}
\label{sec:mwu}

In this section, we show that if we restrict our attention to subset queries, then we can outperform \cref{alg:matching} in terms of approximation ratio. 
We remark that \cref{alg:matching} applies under weaker guarantees, namely those of the $\textsc{ComputeDegrees}_\theta$ subroutine in \cref{sec:vertex-query-implementation}, and can therefore be implemented efficiently using either subset or vertex queries.
This flexibility comes at the cost of a larger approximation ratio.
In contrast, the algorithm in this section requires stronger guarantees (as in \cref{assume:compute-degrees}) that are achievable only through subset queries.

We prove
\begin{theorem}
\label{thm:mwu-subset-query}
    Given an $n$-vertex graph $G$, for any $\epsilon > 0$, there is a randomized algorithm that $(1.5 + \epsilon)$-approximates the size of maximum matching in $G$ using $\poly(\log n, 1/\epsilon)$ subset queries to $G$ w.h.p.
    If the graph is bipartite, then the approximation ratio improves to $1 + \epsilon$.
\end{theorem}
This is achieved through a $(1 + \epsilon)$-approximation of the minimum \emph{fractional} vertex cover, whose value is equal to the maximum fractional matching.
In turn, this yields a $(1.5 + \epsilon)$-approximation for general (i.e., possibly non-bipartite) graphs, due to the $1.5$ integrality gap, and a $(1 + \epsilon)$-approximation for bipartite graphs.

Recall that, for a Boolean function $f$ on a DAG $G$, the conflict graph $G_f$ satisfies $\dmon{f} = \mu(G_f)$ (\cref{fact:matching-dmon}).
Therefore, since $G_f$ is bipartite, and a subset query to $G_f$ can be implemented in $\sqrt{n}^{1 + o(1)}$(\cref{lem:subset-query-streaming}) passes and space $\widetilde{O}(n/\epsilon^2)$, we have:

\begin{theorem}\label{thm:mwu-dmon}
    For any $\epsilon > 0$, there is a randomized streaming algorithm that takes $\sqrt{n}^{1+o(1)} \cdot \poly(1/\eps)$  passes over the stream, uses $\widetilde{O}(n/\eps^2)$ space, and w.h.p. $(1+\eps)$-approximates $\dmon{f}$.
\end{theorem}

The algorithm relies on the \textsc{ComputeDegrees} subroutine, as defined in \cref{assume:compute-degrees}, which produces a $(1+\epsilon)$-approximation of all the vertex degrees using $O\left(\frac{\log^2 n}{\epsilon^3}\right)$ subset queries (\cref{lem:degree-estimation-subset-queries}).
We adapt Lemma 3.5 of Liu~\cite{Liu24}, which implements a multiplicative-weight-update (MWU) oracle using $O(\epsilon^{-2}\log n)$ calls to \textsc{ComputeDegrees}. Then, the value of the minimum fractional vertex cover can be $(1 + \epsilon)$-approximated using $T = O(\epsilon^{-2} \log n)$ calls to the MWU oracle.

\subsection*{Definitions}

We begin with some definitions. Given a graph $G = (V, E)$ a \emph{fractional vertex cover} is a set of values $y \in \R_{\geq 0}^V$, such that $$y_u + y_v \geq 1 \qquad \text{for every }(u,v)\in E.$$
We use $\tau_f(G)$ to denote the value of the minimum fractional vertex cover, which is the minimum possible value of $\norm{y}_1$ among all fractional vertex covers $y$. 

For the MWU framework, we consider parameters $c, \lambda, \epsilon > 0$, and define a potential function $\Phi: \R_{\geq0}^V \to \R$,
$$
\Phi(y) = \sum_{(u, v)\in E}\exp(-\lambda(y_u + y_v)).
$$
For a fixed $y$, we similarly define $s_v$ to be the sum of the terms in the potential corresponding to edges incident to $v$:
$$
s_v = \sum_{u \in N(v)}\exp(-\lambda(y_u + y_v)).
$$

\begin{definition}[MWU oracle] \label{def:mwu-oracle}
    Given a graph $G = (V, E)$, $y \in \R_{\geq0}^V$, and $c, \lambda, \epsilon > 0$, an MWU oracle either fails or outputs an update vector $\Delta \in \R_{\geq 0}^V$ such that:
    \begin{enumerate}
        \item $\sum_v \Delta_v = 1$, \label{item:mwu-oracle-1}
        \item $\sum_v \Delta_v s_v \geq (1 - \epsilon)c \cdot \Phi(y)$, and \label{item:mwu-oracle-2}
        \item $\Delta_v \leq \frac{\epsilon}{2\lambda}$ for all $v \in V$. \label{item:mwu-oracle-3}
    \end{enumerate}
    When $\frac{1}{2\tau_f(G)} \leq c \leq \frac{1}{\tau_f(G)}$, the oracle is guaranteed to succeed.
\end{definition}

\subsection*{The MWU Framework}\label{sec:mwu}

First, we give an overview of how an MWU oracle can be used to obtain a $(1 + \epsilon)$-approximation of the minimum fractional vertex cover.
The reader may refer to \cite[Section 3]{Liu24} for a full proof.
We remark that their proof is stated for $\mu(G)$, the size of the maximum matching, in bipartite graphs.
The only modification in our setting appear is in \cref{clm:mwu-oracle}, where we show how the MWU oracle is obtained from \textsc{ComputeDegrees}.

\begin{claim}[See {\cite[Section 3]{Liu24}}]
    \label{clm:mwu-framework}
    Given a graph $G$, for $\lambda = \epsilon/10c$, a $(1 + \epsilon)$-approximate minimum fractional vertex cover can be computed using $\poly(\log n, 1/\epsilon)$ calls to the MWU oracle.
\end{claim}
\begin{proof}[Proof sketch]
    First, observe that the approximation problem can be reduced to a decision problem as follows:
    For $O(\frac{\log n}{\epsilon})$ choices of $c = (1 + \epsilon)^{-i}$, we test whether there exists a $y \in \R_{\geq 0 }^V$ such that
    \begin{equation}
    \sum_v y_v = 1 \qquad \text{and} \qquad y_u + y_v \geq c, \ \forall(u,v) \in E. \label{eq:mwu-decision}
    \end{equation}
    This is feasible if and only if $c \leq 1/\tau_f(G)$. Therefore, the largest feasible guess yields a $(1 + \epsilon)$-approximation of $\tau_f(G)$.

    On a high level, the decision problem is solved as follows.
    Let $\lambda = \epsilon/10c$ and $T = 20\epsilon^{-2}\log n$.
    Starting with $y^{(0)} = 0$, for $t = 0, 1, \ldots, T - 1$, let 
    $y^{(t+1)} = y^{(t)} + \Delta^{(t)}$, where $\Delta^{(t)}$ is the output of the MWU oracle for $y^{(t)}$.
    It can be shown that when the oracle succeeds in every step, 
    $\overline{y} = y^{(T)}/T$ satisfies \eqref{eq:mwu-decision} up to a $1 + O(\epsilon)$ factor.  
    Hence, rescaling $\overline{y}$ gives a fractional vertex cover of size $(1 + O(\epsilon))\frac{1}{c}$.
    As a result, considering the largest $c$ for which the oracle succeeds in every step, yields the $(1 + O(\epsilon))$-approximation.
    We remark that the oracle may behave arbitrarily when $c < \frac{1}{2\tau_f(G)}$.
\end{proof}

Therefore, it suffices to show that an MWU oracle can be obtained using \textsc{ComputeDegrees}. We make the necessary adaptations for general graphs and $\tau_f(v)$.

\begin{claim}[Analogous to {\cite[Lemma 3.5]{Liu24}}]
    \label{clm:mwu-oracle}
    The MWU oracle, as defined in \cref{def:mwu-oracle}, can be implemented using $O(\epsilon^{-2} \log n)$ calls to \textsc{ComputeDegrees}.
\end{claim}
\begin{proof}
    First, we obtain estimates $\tilde{s}_v$ for $s_v$ for all vertices $v$:
    \begin{equation}
    (1 - \epsilon / 2) s_v \leq \tilde{s}_v \leq  s_v, \label{eq:mwu-oracle-s}
    \end{equation}
    where recall that 
    $$s_v = \sum_{u \in N(v)} e^{-\lambda(y_u + y_v)}
    = e^{-\lambda y_v} \sum_{u \in N(v)} e^{-\lambda y_u}.
    $$
    To do so, we partition the vertices $V$ into sets $V_0, \ldots, V_{5T}$
    where
    $$
    V_i = \{v \in V\mid i\epsilon / 10 \leq \lambda y_v < (i+1)\epsilon/10\}.
    $$
    Note that $5T$ suffices as the maximum index, since by the construction of $y$, we have $y_v \leq T \cdot \frac{\epsilon}{2\lambda}$ (due to \Cref{item:mwu-oracle-3} of \cref{def:mwu-oracle}).
    Then, we estimate $\sum_{u \in N(v) \cap V_i} e^{-\lambda y_u}$ by $\card{N(v) \cap V_i} e^{-i\epsilon /10}$,
    where $\card{N(v) \cap V_i}$ can be estimated simultaneously for all $v$, up to a $(1 + \epsilon)$ factor using \textsc{ComputeDegrees}.
    Formally, letting $\tilde{d}_{v,i}$ be the $(1 + \epsilon)$-approximation for $\card{N(v) \cap V_i}$, we have:
    \begin{align*}
    \sum_i \tilde{d}_{v, i} e^{-i \epsilon/10}
    &= (1 \pm O(\epsilon)) \sum_i \card{N(v) \cap V_i} e^{-i \epsilon/10}\\
    &= (1 \pm O(\epsilon)) \sum_i \sum_{u \in N(v) \cap V_i} e^{-\lambda y_u}\\
    &= (1 \pm O(\epsilon)) \sum_{u \in N(v)} e^{-\lambda y_u}.
    \end{align*}
    Multiplying this value by $\frac{e^{-\lambda y_v}}{1 + O(\epsilon)}$ and rescaling $\epsilon$ appropriately, yields the desired estimate $\tilde{s}_v$.

    Next, we use these estimates $\tilde{s}_v$ to construct $\Delta$.
    Sort the vertices $v_1, v_2, \ldots, v_n$ in decreasing order of $\tilde{s}_v$.
    Let $U$ be the largest prefix that satisfies:
    \begin{equation}
    \frac{1}{\card{U}}\sum_{v\in U}\tilde{s}_v
    \geq (1-\epsilon)c\Phi(y). \label{eq:mwu-oracle-u}
    \end{equation}
    Then, let
    $$
    \Delta_v = \begin{cases}
        1 / \card{U}, & \text{if } v \in U, \text{ and} \\
        0, & \text{otherwise.}
    \end{cases}
    $$

    Observe that \Cref{item:mwu-oracle-1} of \Cref{def:mwu-oracle} is satisfied trivially by the choice of $\Delta$, and \Cref{item:mwu-oracle-2} is satisfied due to \eqref{eq:mwu-oracle-s} and \eqref{eq:mwu-oracle-u}.
    When $\card{U} < \frac{2\lambda}{\epsilon}$, the algorithm declares failure.
    When $c \in \left[\frac{1}{2\tau_f(G)}, \frac{1}{\tau_f(G)}\right]$, we show that $\card{U} \geq \frac{2\lambda}{\epsilon}$ and hence \Cref{item:mwu-oracle-3} holds. 

    We show that $\card{U} \geq \floor{\tau_f(G)}$.
    In turn, it follows that
    $$
    \card{U} \geq \frac{\tau_f(G)}{2} \geq \frac{1}{4c} \geq \frac{2\lambda}{\epsilon}.
    $$
    We argue that the prefix of size $k = \floor{\tau_f(G)}$ satisfies \eqref{eq:mwu-oracle-u}.
    Let $x \in \R^V_{\geq 0}$ be an optimal fractional minimum vertex cover.
    Consider $\sum_v x_v\tilde{s}_v$ as a weighted sum of $\{\tilde{s}_v\}_v$.
    Since $x_v \leq 1$ for all $v$, the sum can only grow, by concentrating the weight on the $k$ vertices with the largest $\tilde{s}_v$.
    That is, giving a weight of $\frac{\tau_f(G)}{k}$ to the first $k$ vertices:
    \begin{equation}
    \sum_v x_v\tilde{s}_v \leq \frac{\tau_f(G)}{k} \sum_{i=1}^k \tilde{s}_{v_i}. \label{eq:mwu-oracle-1}
    \end{equation}
    On the other hand, using the fact that $x$ is a vertex cover, we get:
    \begin{align}
    \sum_vx_v\tilde{s}_v
    &\geq(1-\epsilon/2)\sum_vx_vs_v \notag\\
    &=(1-\epsilon/2)
    \sum_{(u,v)\in E}e^{-\lambda(y_u+y_v)}(x_u+x_v) \notag\\
    &\geq(1-\epsilon/2)\Phi(y). \label{eq:mwu-oracle-2}
    \end{align}

    Putting \eqref{eq:mwu-oracle-1} and \eqref{eq:mwu-oracle-2} together, we get:
    $$
    \frac{1}{k}\sum_{i=1}^k\tilde{s}{v_i}
    \geq \frac{1-\epsilon/2}{\tau_f(G)}\Phi(y)
    \geq (1-\epsilon/2)c\Phi(y).
    $$
    Therefore, $U \geq \floor{\tau_f(G)}$, which concludes the proof of \Cref{item:mwu-oracle-3}.
    We finally remark that, since there are $5T = O(\epsilon^{-2}\log n)$ vertex sets, the oracle requires as many calls to \textsc{ComputeDegrees} for one guess $c$.
\end{proof}

To conclude, we remark that \cref{clm:mwu-framework,clm:mwu-oracle} together imply that a $(1 + \epsilon)$-approximation of the minimum fractional vertex cover can be computed using $\poly(\log n, 1/\epsilon)$ subset queries.
Hence, \cref{thm:mwu-subset-query} follows which, in turn, implies \cref{thm:mwu-dmon}.

\section{The Lower Bound}\label{sec:lb-reach}

In this section, we establish a lower bound on the pass complexity of approximating the distance to monotonicity in the streaming model,
which is stated in the following theorem. 

\begin{theorem}\label{thm:lb}
    For any fixed $c \geq 1$ there is $\epsilon > 0$ such that a $(c, \epsilon n)$-approximation of $\dmon{f}$ with $\widetilde{O}(n)$ space requires $\Omega(\streach)$ passes, where $\streach$ denotes the pass-complexity of the best $\widetilde{O}(n)$ space streaming algorithm for the $st$-reachability problem.
\end{theorem}

Despite extensive studies, all we know about $\streach$ is that
$$
    \Omega\left(\frac{\log n}{(\log \log n)^2}\right) \stackrel{\text{\cite{CKPSSY21}}}{\leq} \streach \stackrel{\text{\cite{JLS19}}}{\leq} n^{1/2+o(1)}.
$$
Note that \cref{thm:lb} can also be seen as a conditional lower bound. It implies that \cref{thm:main} has optimal pass-complexity modulo improving the state-of-the-art $st$-reachability algorithm of \cite{JLS19}.


\begin{figure}[h]
    \centering
    
    \includegraphics[width=0.8\textwidth]{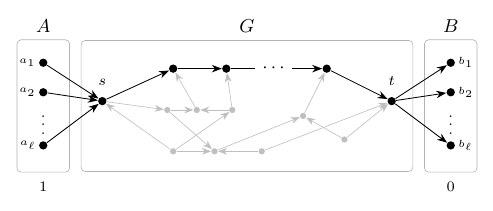}

    \caption{ If $t$ is reachable from $s$ the instance is far from monotone.
}
    \label{fig:reach}
\end{figure}

Let $G = (V, E)$ be a directed graph with $|V| = n$, $|E| = m$, and vertices
$s, t \in V$ be an instance of $st$-reachability. 
  We construct a directed graph $G' = (V', E')$ and a boolean
function $f \colon V' \to \{0,1\}$ as follows.
Let $l= 2cn$, for a fixed constant $c$.
Introduce two fresh batches of vertices
\[
    A = \{a_1, \dots, a_l\}, \qquad B = \{b_1, \dots, b_l\},
\]
disjoint from $V$ and from each other, and set $V' = A \cup V \cup B$
with $|V'| = 2l + n$.  The edge set $E'$ consists of
  $(a, s)$ for every $a \in A$, 
  $(u, v)$ for every $(u,v) \in E$ and,
 $(t, b)$ for every $b \in B$.
Finally, define $f \colon V' \to \{0,1\}$ by
\[
    f(v) =
    \begin{cases}
        1 & v \in A,\\[2pt]
        0 & v \in B,\\[2pt]
        0 & v \in V.
    \end{cases}
\]

First, we investigate the distance to monotonicity in the following two cases. The first case is when $t$ is reachable from $s$, we call this the yes case, represented in \cref{fig:reach}.  The no case is when $t$ is not reachable from $s$.

\begin{lemma}\label{lem:yes-case} 
If $s$ can reach $t$ in $G$, then $\dmon{f} \ge 2cn$. 
\end{lemma}

\begin{proof}
Fix a directed path $s$ to $t$ in $G$.
For every $(a, b) \in A \times B$, there is a directed path $ a $ to $ b$ in $G'$ via the path from $s$ to $t$, so $a$ can reach $b$ while $f(a) = 1 > 0 = f(b)$.
All $l^2$ pairs in $A \times B$ are violations.
 
Any monotone function $g$ must fix at least one endpoint per violation.
Since these violations form the complete bipartite graph, any
vertex cover has size at least $l$ , so in $G'$
$\dmon{f} \ge l = 2cn$.
\end{proof}

\begin{lemma}\label{lem:no-case}
If $s$ cannot reach $t$ in $G$, then $\dmon{f} \le n$.
\end{lemma}

\begin{proof}
Let $g$ be a monotone function that agrees with $f$ on all but at
most $n$ vertices.  Let $R_s \subseteq V$ be the set of vertices
reachable from $s$ in $G$ (including $s$ itself).  Define $g$ as follows:
\[
    g(v) =
    \begin{cases}
        1 & v \in A \cup R_s,\\[2pt]
        0 & \text{otherwise}.
    \end{cases}
\]
We verify that $g$ is monotone on $G'$, i.e., for every edge $(x,y) \in E'$
we have $g(x) \le g(y)$.  
For edges $(a, s)$ where $a \in A$ we have $g(a) = 1$ and $g(s) = 1$ since $s \in R_s$.
For edges $(u,v) \in E$ where $u \in R_s$, we have $v \in R_s$ as well,
so $g(u) = g(v) = 1$.
For edges $(u,v) \in E$ where $u \notin R_s$ we have $g(u) = 0 \le g(v)$.
For edges $(t, b)$ where $b \in B$, since $s$ cannot reach $t$, we have $t \notin R_s$,
so $g(t) = 0 = g(b)$.

Therefore, the function $g$ differs from $f$ only on $R_s \subseteq V$ where $f = 0$
but $g = 1$, so $\dmon{f} \le |R_s| \le n$.
\end{proof}



\begin{proof}[Proof of \cref{thm:lb}]
Set $\epsilon = \frac{c}{2(4c+1)}$, which satisfies $\epsilon > 0$ for any fixed $c \ge 1$. Suppose $\mathcal{A}$ is a streaming algorithm that computes a $(c, \epsilon N)$-approximation of $\dmon{f}$ on posets of size $N$ in $p$ passes using space $\widetilde{O}(N)$.  The edges in $G$ arrive in the same sequence as the stream for the $s$-$t$ reachability problem. Assume $f$ arrives in the beginning of the sream.

Let $x$ be the value returned by running $\mathcal{A}$ on the stream for $G'$. 
If $t$ is reachable from $s$, by \cref{lem:yes-case}, $\dmon{f} \ge 2cn$. Since $\mathcal{A}$ is a $(c, \epsilon N)$-approximation, its output must satisfy $x \ge \dmon{f} \ge 2cn$.

If $t$ is not reachable from $s$, by \cref{lem:no-case}, $\dmon{f} \le n$. The algorithm $\mathcal{A}$ must output a value $x \le c \cdot \dmon{f} + \epsilon N$. Since $N = (4c+1)n$, given our choice of $\epsilon$ we have $\epsilon N = \frac{c}{2(4c+1)}(4c+1)n = 0.5cn$. Thus, $x \le cn + 0.5cn = 1.5cn$.

Because $1.5cn < 2cn$, the output of $\mathcal{A}$ allows us to perfectly distinguish whether $s$ can reach $t$ in $G$. Since $N = \Theta(n)$ and $|E'| = m + 2l = m + \Theta(n)$, generating the stream for $G'$ requires only $\widetilde{O}(n)$ space. Therefore, $\mathcal{A}$ effectively solves $st$-reachability in $p$ passes and $\widetilde{O}(n)$ space, implying that $p = \Omega(\streach)$.
\end{proof}

\section*{Acknowledgments}

We thank an anonymous reviewer for pointing out the alternative approach detailed in Section~\ref{sec:mwu}.

\bibliographystyle{alpha}

\bibliography{references}

\end{document}